%% file: contentious-paper.tex
\documentclass[]{article}

\usepackage[a4paper]{geometry}

\usepackage{fourier}

\usepackage{url}

\usepackage{algorithm}
\usepackage{algpseudocode}
\usepackage{listings}
\usepackage{booktabs}
\usepackage{diagbox}
\usepackage{amsmath}
\usepackage{amssymb}
\usepackage{amsfonts}
\usepackage{xspace}
\usepackage{enumitem} \setlist[itemize]{noitemsep} \setlist[enumerate]{noitemsep}
\usepackage{color}
\usepackage{graphicx}
\usepackage{multirow}
\usepackage{authblk}
\usepackage{amsthm}
\usepackage{caption} \usepackage{subcaption}
\usepackage{balance}
\usepackage{thmtools}
\usepackage{thm-restate}

\newtheorem{theorem}{Theorem}
\declaretheorem{example}
\renewcommand\thmcontinues[1]{continued}
\newtheorem{fact}{Fact}
\newtheorem{proposition}{Proposition}
\newtheorem{lemma}{Lemma}
\newtheorem{corollary}{Corollary}
\newtheorem{definition}{Definition}
\newtheorem{remark}{Remark}

\newcolumntype{C}[1]{>{\centering\let\newline\\\arraybackslash\hspace{0pt}}m{#1}}

\newcommand\pcase[1]{}

\newcommand{\Gexone}[1]{\tilde{G}_{\ref{ex:bringyourownfood}#1}}
\newcommand{\Gextwo}{\tilde{G}_{\ref{ex:telecommunication}}}

\newcommand{\Gexcont}{\tilde{G}_c}
\newcommand{\Gexrcgbarbyof}{\tilde{G}_{\ref{ex:rcgbar-bringyourownfood}}}
\newcommand{\Gexrcgbargf}{\tilde{G}_{\ref{ex:grandfather}}}

\newcommand{\var}[1]{\mathsf{#1}}
\renewcommand{\phi}{\varphi}

\newcommand{\Card}[1]{{\sf Card}(#1)}
\newcommand{\VAR}{{\sf var}}

\newcommand{\pref}{{\sf pref}}
\newcommand{\pct}{{\sf ct}}
\newcommand{\pca}{{\sf ca}}
\newcommand{\ppca}{{\sf pca}}

\newcommand{\Res}{{\ensuremath{Res}\xspace}}
\def\addto#1#2{\expandafter\def\expandafter#1\expandafter{#1#2}}
\def\replacestrings#1#2{%
   \def\tmp##1#1##2\end{\ifx\end##2\end\addto\tmpb{##1}\else\addto\tmpb{##1#2}\tmp##2\end\fi}%
   \expandafter\def\expandafter\tmpb\expandafter{\expandafter}\expandafter\tmp\tmpb#1\end 
}

\makeatletter
\newcommand*\bigcdot{\mathpalette\bigcdot@{.5}}
\newcommand*\bigcdot@[2]{\mathbin{\vcenter{\hbox{\scalebox{#2}{$\m@th#1\bullet$}}}}}
\makeatother

\newcommand{\comb}{\bigcdot}

\newcommand{\emptybundle}{\lambda}
\newcommand{\transform}{\rightsquigarrow}
\def\bundle#1{\def\tmpb{#1}\replacestrings{,}{\comb}\tmpb}
\newcommand{\bag}[1]{{\{#1\}\xspace}}
\newcommand{\ch}{\mathsf{ch}}
\newcommand{\out}{\mathsf{out}}

\newcommand{\SAT}{{\sf 3\text{-}SAT}\xspace}

\newcommand{\NE}{{\sf NE}\xspace}
\newcommand{\ENE}{{\ensuremath{\sf \exists NE}}\xspace}
\newcommand{\EPD}{{\ensuremath{\sf \exists PD}}\xspace}
\newcommand{\ISAT}{{\ensuremath{\mathop{\mbox{\sf psat-prof}}}}\xspace}
\newcommand{\SATP}{{\ensuremath{\mathop{\mbox{\sf psat-play}}}}\xspace}
\newcommand{\GOOD}{{\ensuremath{\mathop{\mbox{\sf good-prof}}}}\xspace}
\newcommand{\CONT}{{\ensuremath{\mathop{\mbox{\sf cont-prof}}}}\xspace}

\newcommand\contenders{\mathsf{Contenders}}
\newcommand\miss{\mathsf{miss}}
\newcommand\deficit{\mathsf{deficit}}
\newcommand\contGoals{\mathsf{ContRes}}
\newcommand\Deficit{\mathsf{Deficit}}

\newcommand{\co}{{\ensuremath{\textsc{co}}}\xspace}
\newcommand{\PTIME}{{\ensuremath{\textsc{PTime}}}\xspace}
\newcommand{\NP}{{\ensuremath{\textsc{NP}}}\xspace}

\newcommand{\coNP}{{\ensuremath{\textsc{coNP}}}\xspace}

\newcommand\takeaway[1]{}

\usepackage[defaultlines=3,all]{nowidow}

\title{ Existence and Verification of Nash Equilibria in\\
  Non-Cooperative Contribution Games\\
  with Resource Contention
}
\author{Nicolas Troquard} 
\affil{

Gran Sasso Science Institute (GSSI)\\
Viale Luigi Rendina, 26-28\\
67100 L'Aquila, Italy\\

\texttt{nicolas.troquard@gssi.it}\\
}
\date{}

\newif\ifshort
\shortfalse
\newif\ifdc
\dcfalse

\begin{document}

\maketitle

\begin{abstract}
  %
  In resource contribution games, a class of non-cooperative games, the players want to obtain a bundle of resources and are endowed with bags of bundles of resources that they can make available into a common for all to enjoy. Available resources can then be used towards their private goals.

  A player is potentially satisfied with a profile of contributed resources when his bundle could be extracted from the contributed resources.
Resource contention occurs when the players who are potentially satisfied, cannot actually all obtain their bundle.
  The player's preferences are always single-minded (they consider a profile good or they do not) and parsimonious (between two profiles that are equally good, they prefer the profile where they contribute less). 
  What makes a profile of contributed resources good for a player depends on their attitude towards resource contention.

We study the problem of deciding whether an outcome is a pure Nash equilibrium for three kinds of players' attitudes towards resource contention: public contention-aversity, private contention-aversity, and contention-tolerance.

In particular, we demonstrate that in the general case when the players are contention-averse, then the problem is harder than when they are contention-tolerant. We then identify a natural class of games where, in presence of contention-averse preferences, it becomes tractable, and where there is always a Nash equilibrium.
\footnote{This version differs from the one published in \cite{journals/amai/Troquard24} in two main parts. First, it provides a counterexample to the conjecture page~\pageref{conjecture} stating that the best-response dynamics can be applied from any profile of RCGBARs with public contention-averse preferences.
  Second, the proof of Lemma~\ref{lem:algo-RCGBAR-EPD-empty} published in \cite{journals/amai/Troquard24} contained a typo. After reconsideration, it also contained many gaps. The proof presented here provides more details.}
\end{abstract}


\section{Introduction}

Most real-world agents in social and socio-technical environments consume resources: money,\linebreak amount of matter, etc. Resources can be transformed, split and recombined, and they constitute a driving force towards the end goal of obtaining resource goods.

\emph{Generalized exchanges}~\cite{ekeh74}
are social settings where the participants are providers and receivers of resources. They involve indirect reciprocity between participants.
%
When there exists no fixed structure and the participants essentially pool the resources to be used by everyone, the exchanges are called \emph{group-generalized}~\cite{yamagishiCook93}.
Social interactions of this sort are common-place: 
mutual help in a neighborhood, the redistribution of goods within a close-knit group, etc.

In~\cite{Tr16ijcai}, a model of interaction between self-interested agents is presented in which the agents are producers and consumers of resources.
When the amount of some wanted resources surpasses the amount of these resources that is produced, we are in presence of resource contention.
We thus take up where~\cite{Tr16ijcai} left off. We introduce
novel kinds of preferences that reflect the players' aversion towards resource contention
in a class of non-cooperative resource games.

The models and the algorithms presented here can be used as analytical tools at the disposition of actors and policy makers.
They can serve at gauging the possible strategic behaviours of the actors and of their competitors, and at identifying possible issues of resource scarcity in a common.
This paper does not however intend to propose possible mechanisms to achieve `good' behaviours in non-cooperative resource games.

%

\paragraph{Voluntary provisions to a common.}
Specifically, the models can be conceived as a common~\cite{Ostrom90} whose resources are voluntarily contributed by the players.
The players are endowed with a multiset of resource bundles, that cannot be used directly, but can be added into a common pool of resources. A possible action of a player consists in providing a submultiset of his endowment to the common pool.
Each player then has a resource bundle objective which he 
can try to attain by using the resources contributed by all the players.
Once in the common pool, the resources are thus \emph{common pool resources}: they are \emph{non-excludable} (every player can consume them) and \emph{rivalrous} (one player's consumption of a resource limits the access of it from other players).
We call these games, \emph{resource contribution games} (RCGs).
%

\medskip

The resources contributed to the common pool should be reminiscent of \emph{voluntary private provisions}~\cite{sugden84,MARKISAAC198551,BERGSTROM198625,Bagnoli91} of \emph{public goods} (e.g., public street lighting, public radio). 
However, public goods are \emph{non-rivalrous}, while in RCGs, the
private provisions are literally consumed towards an agent's private goal, leaving less resources for the other agents to consume.
The configuration of RCGs may resonate more closely with the setting adopted in the \emph{public good game}~\cite{MARWELL1981295}, a standard game of experimental economics. In the public good game, the participants are given a number of private tokens, and an action is to put some number (possibly zero) of these tokens in a public pot. The participants get to keep the tokens that they did not contribute. Then the tokens in the pot are multiplied by a factor and distributed evenly among the participants.
However, the goals of the participants of the public good game are only monetary and the redistribution of common pool resources is a simple division of goods by a supervising authority.

In RCGs, the goals of the players are arbitrary bundles of resources that the players can only satisfy by using the common pool resources without the supervised resource allocation of an authority. 

\begin{example}[label=ex:bringyourownfood]
  You participate in a bring-your-own-food cooking party. You bring a bottle of white wine. You cannot take it out and open it without sharing it, and if other guests also want a taste of it, or to cook with it, you might not have enough to your satisfaction. Once at the party you can strategically decide whether to contribute the bottle to the party or keep it away in your bag.
  Your bottle of wine contains four glasses, and your objective is to enjoy two glasses.
  Bruno and Carmen both want a portion of risotto, each portion needing a portion of rice, one onion, and a glass of white wine. Bruno brings three portions of rice, and Carmen brings two onions.

  What are the possible equilibria?
  \end{example}

Unsupervised resource-sensitive social interactions like the one in Example~\ref{ex:bringyourownfood} have attracted less attention in the literature than public good games.
Resource contribution games are not either elaborated models of interaction with particularly pleasing properties, especially in their most general form.
The situations they capture are not any less ubiquitous, at least in our mundane everyday life, and are well worth studying for what they are.
They are arguably rarer in the economic sector, but occur for instance in sharing economies~\cite{doi:10.1177/000276427802100411} or interconnected economies~\cite{brock1995interconnection,armstrong98interconnection}. 
The following example, taken from~\cite{10.1093/logcom/exaa031}, is a simple abstraction of the forces in presence in an interconnected market in the telecom industry.
\begin{example}[label=ex:telecommunication]
  In a local telecom industry, anti-trust laws forbid a priori cooperation, and regulations oblige the companies to accept traffic from each other.
  Consider two competing telecommunication companies. 
  Company~$A$ manages a 3G network of comprised capacity~$4$ (bundled as capacities~$1$, $1$ and $2$). Company~$B$ manages a 3G network of capacity~$1$, and a 4G network of comprised capacity~$3$ (bundled as capacities~$1$ and $2$).
  Company~$A$ needs to offer its customers 3G at capacity~$2$ and 4G at capacity~$1$. Company~$B$ needs to offer its customers 3G at capacity $2$ and 4G at capacity~$2$.

  Companies have preferences along these lines: (1)~Activating a network at some capacity has a cost; Contributing less is better.
  (2)~Serving under-capacitated network to customers may yield to various technical and economic failures; having access to less than what they need is not acceptable.
  (3)~Con\-tention-averse companies will not accept using oversubscribed networks.
  %

  What are the possible equilibria?
\end{example}

We will formalize Example~\ref{ex:bringyourownfood} and Example~\ref{ex:telecommunication} later and use them to illustrate the notions introduced in this paper.

\paragraph{The need for contention-averse preferences.}
In~\cite{Tr16ijcai}, a model of non-cooperative games is proposed for representing situations such as the one in Example~\ref{ex:bringyourownfood}. 
There are two Nash equilibria. In both, you share the bottle of wine. In the first one, Bruno and Carmen do not share anything: only you are satisfied, but Bruno (resp.\ Carmen) cannot cook the risotto without Carmen's onions (resp.\ Bruno's rice), and thus has no individual incentive to contribute.

In the second one, Bruno contributes one portion of rice, and Carmen contributes one onion. They can now cook one portion of risotto. They are both `satisfied', and 
they have no incentive to contribute more, that would be necessary to cook a second portion of risotto. They are both `satisfied' by the fact that a portion of risotto can be cooked, and they are not concerned with an inevitable contention about who is actually going to eat the sole portion of risotto.

To avoid this second often unsatisfactory solution, we add contention-aversity to the parsimonious preferences. In this example, it will yield two Nash equilibria. The first one is as before. In the second one, you still contribute the wine, Bruno contributes two portions of rice and Carmen contributes two onions; enough to cook two portions of risotto and for you to enjoy enough of your bottle of wine.

\begin{example}[continues=ex:bringyourownfood]
Suppose Edward also comes to the party without bringing anything, and wants to enjoy three glasses of white wine.
\end{example}
With the original so-called \emph{contention tolerant} preferences, the profile where you contribute the bottle of wine, Bruno contributes one portion of rice, Carmen contributes one onion, and Edward contributes nothing is a Nash equilibrium, a very contentious one.
With the new so-called \emph{contention-averse} preferences introduced in this paper, the only Nash equilibrium is when no-one contributes anything.

\paragraph{Attitudes towards resources and preferences.}

When a player can obtain his objective from the \linebreak pool of contributed resources, we will say that he is \emph{potentially satisfied}.
An action profile providing enough resources to satisfy potentially a set of players but not enough to satisfy them all at once is a \emph{contentious profile}. \emph{Contention-tolerant} players do not worry about contentious profiles, and find them \emph{good}, as long as they can be potentially satisfied. \emph{Contention-averse} players will not consider a contentious profile as good, even if they can be potentially satisfied. We will call these players \emph{public contention-averse}.
A contention-averse attitude is specifically pertinent in close-knit groups where members, although non-cooperative, are community-conscious and suffer from any conflict over resources. A bring-your-own food party is a real-world example where contention about limited resources affects all the participants, and must be avoided at all costs. (Like, e.g., in computing, resource-contention leads to thrashing that affects all processes. In geopolitics, three neighboring countries with interconnected economies would prefer obtaining the goods they need and avoid resource contention, which could possibly lead to conflict.)
In less interdependent societies, a resource contention in a profile may not concern a player directly. It happens when the contention is about a resource that is independent of the player's objective. \emph{Private contention-averse} players that are potentially satisfied in such a profile will consider it good. Of course, they will not consider good a profile in which the contention is about a resource that plays a part in their objective.

The players are \emph{single-minded} (as, e.g., in~\cite{aziz19,DBLP:conf/icalp/KeijzerKV20}): either they find a profile good, or they do not. A player will prefer a good profile from a `not good' profile.

In addition to this, players are \emph{parsimonious}. They prefer a profile in which they contribute strictly less (for multiset inclusion) from another profile, if they find the profile otherwise equally good. Analogous preferences have sometimes been called \emph{pseudo-dichotomous}. It has been used before in the literature, e.g., \cite{WOOLDRIDGE2013418,DBLP:journals/jair/GrandiGT19}.
Pseudo-dichotomous preferences are a simple tool to add nuance to how the players regard the outcomes. In the context of resource contributions, they are a way to disincentivize waste.
The best outcome for a player is to be in a profile that he considers good and in which he does not contribute anything. A worst outcome is to be in a profile that he does not consider good and in which he contributes all of his endowment.

To summarize, the players are single-mindedly pursuing profiles that are good for them and they have pseudo-dichotomous preferences. To be good for a player, a profile must potentially satisfy the player, but it also depends on the player's exact attitude towards contention.

\takeaway{ 
\paragraph{Succinct games.}
The actions of a player the submultisets of his endowment, and his objective is described by a bundle of resources. This makes the game specification a \emph{succinct} representation of of the scenario being modelled. They can be reminiscent of other succinct game representations.
In boolean games~\cite{Harrenstein:2001:BG:1028128.1028159,DBLP:conf/ecai/BonzonLLZ06}, each player controls a set of boolean variables and produces truth values which can be used without restriction towards the boolean goal of all the players. As such, resource contention is absent from boolean games.
In congestion games~\cite{Rosenthal1973} and in the compact singleton congestion games~\cite{ieong05}, the players choose a set of resources to use, and their utility depends on the `delays' of shared resources, which depend on the number of players choosing them. Although delays can be arbitrarily large, the resources do not become oversubscribed, and they are not subject to contention \emph{per se}.
Unlike in congestion games and boolean games, the atomic resources considered in this paper are not dedicated to be shared, and the oversubscription of some type of resources may become a critical point of contention.
} 

\paragraph{Outline.}
Our main conceptual contribution is the definition of the notion of resource contention in RCGs, and the definition of contention-averse preferences. Our main technical contributions will be on the computational aspects of Nash equilibria in this setting.

The class of resource contribution games (RCGs) is introduced in Section~\ref{sec:irg-dp}. We define resource-contention and three kinds of preferences, namely, public and private contention-averse preferences and contention-tolerant preferences. We also define the notion of pure Nash equilibrium.
We focus on the problem of deciding whether a profile is a (pure) Nash equilibrium (\NE), and also address interesting cases of the problem of deciding whether a game admits a (pure) Nash equilibrium (\ENE).
In Section~\ref{sec:case-contention-aversity}, we show that the problem \NE is $\coNP$-complete in presence of public and private contention-averse preferences.
In Section~\ref{sec:case-contention-tolerance}, we show that the problem is in $\PTIME$ in presence of contention-tolerant preferences. We also show that the problem of deciding whether a Nash equilibrium exists in a game with contention-tolerant preferences is $\NP$-complete.
In Section~\ref{sec:irgbar}, we focus on a class of games where the endowments are bags of atomic resources (RCGBARs). We show that the problem of deciding whether a profile is a Nash equilibrium is in \PTIME. We also show that in this class of games, if in presence of public contention-averse preferences, a Nash equilibrium always exists, and finding one can be done in polynomial time. This sets public contention-averse preferences apart, as a Nash equilibrium needs not exist in RCGBARs in presence of either contention-tolerant preferences or private contention-averse preferences.
In Section~\ref{sec:case-other-variants} we investigate some more tractable variants of the classes of RCGs.
%
Concluding remarks are offered in Section~\ref{sec:conclusions}, where we discuss related work, open problems, and future work.

We illustrate the use of an implementation of our models and algorithms in the Appendix. 
We provide detailed software executions of Example~\ref{ex:bringyourownfood}, and Example~\ref{ex:telecommunication}, and repeat a few proofs of the paper.

\section{Preliminaries}
\label{sec:irg-dp}

We present the main definitions and some basic results.

\subsection{Multiset notations}

We use $\uplus$ and $\cup$ for multiset disjoint union (sum) and set union respectively. We use $\cap$ for set or multiset intersection. We use $\subseteq$ for set or multiset containment. We use $\setminus$ for set or multiset difference. We use $\subseteq$ for set or multiset inclusion. We use $\in$ for set or multiset membership.

Let $U$ be a set of multiset elements,
and let $X$ and $Y$ be two multisets over a universe $U$.
Let $m(X, x)$ be the \emph{multiplicity} of $x$ in $X$, that is, the number of times $x$ occurs in $X$.
$X$ is the empty multiset when $m(X, x) = 0$ for every $x \in U$.
$X \uplus Y$ is the multiset such that $m(X\uplus Y, x) = m(X, x) + m(Y, x)$, for every $x \in U$.
$X \cap Y$ is the multiset such that $m(X\cap Y, x) = \min(m(X, x), m(Y, x))$, for every $x \in U$.
$X \setminus Y$ is the multiset such that $m(X\setminus Y, x) = \max(0, m(X, x) - m(Y, x))$, for every $x \in U$.
We have $X \subseteq Y$ iff $m(X,x) \leq m(Y,x)$, for every $x \in U$.
We have $x \in X$ iff $m(X, x) > 0$.


\subsection{Bundles, bags, resource transformations, and notation}

We assume a resource domain, which is a set $\Res$ of atomic (indivisible) resource types.
Bundling is a common feature in economic transactions and exchanges. They can be homogeneous (e.g., a box of 12 cookies), or heterogeneous (a 3-DVD box set of a trilogy). Players may be inclined to provide bundles instead of individual items for various reasons (e.g., wholesale vs.~retail sales pricing, cargo shipping).
\medskip

In this paper, resource bundles are multisets over the universe $U = \Res$, and resource bags are multisets over the universe $U$ being the set of resource bundles.

A \emph{resource bundle} is a 
multiset of atomic resources of any type in $\Res$.
For example, let $\var{H}$ and $\var{O}$ be two atomic resource types in $\Res$. The object $\bundle{\var{H},\var{H},\var{O}}$ is the resource bundle containing two instances of the atomic resource $\var{H}$ and one instance of the atomic resource $\var{O}$.
The empty resource bundle containing no instance of any resource type is denoted by $\emptybundle$.
The singleton resource bundle containing only one instance of one resource type $A \in \Res$ is simply denoted by $\bundle{A}$.
Given two resource bundles $Bundle_1 = \bundle{A_1, \ldots, A_k}$ and $Bundle_2 = \bundle{A_{k+1}, \ldots, A_l}$, we denote by $Bundle_1 \comb Bundle_2$ the resource bundle $\bundle{A_1, \ldots, A_l}$.
Let $Bundle = \bundle{A_1, \ldots, A_k}$ be a resource bundle. We define the multiset of atomic resource instances in $Bundle$ as $\flat(Bundle) = \{A_1, \ldots, A_k\}$.

\medskip
A \emph{resource bag} is a multiset of resource bundles.
For example, $\bag{ \bundle{\var{O},\var{O}}, \bundle{\var{H},\var{H},\var{O}}, \bundle{\var{H},\var{H},\var{O}}, \var{C}}$ is the resource bag containing one resource bundle $\bundle{\var{O},\var{O}}$, two resource bundles $\bundle{\var{H},\var{H},\var{O}}$, and one singleton resource bundle $\bundle{\var{C}}$.
Let $Bag = \{B_1, \ldots, B_k\}$ be a resource bag. We define the multiset of atomic resource instances in $Bag$ as $\flat^\bullet(Bag) = \biguplus_{1 \leq i \leq k} \flat(B_i)$.

\medskip
We say that the resource bag $Bag$ \emph{can be transformed into} the resource bundle $Bundle$, noted \[Bag \transform Bundle \enspace ,\] when $\flat(Bundle) \subseteq \flat^\bullet(Bag)$. We simply write $Bag \not\transform Bundle$ when $Bag$ cannot be transformed into $Bundle$.

By definition, resources can be disposed freely during transformation.
%
For example, $\{\var{H} \comb \var{H}, \var{O} \comb \var{O}\} \transform \var{H} \comb \var{H} \comb \var{O}$.

\medskip

We use ${Bundle_1}_{\mid Bundle_2}$ to represent the `sub-bundle' of $Bundle_1$ restricted to resources present in the bundle $Bundle_2$. E.g., ${(A \comb A \comb B)}_{\mid A \comb C} = A \comb A$; or $(A \comb B)_{\mid C} = \emptybundle$ (that is, the empty resource bundle). Formally, ${Bundle_1}_{\mid Bundle_2}$ is the resource bundle
such that
\[
m(\flat({Bundle_1}_{\mid Bundle_2}), x) =
\begin{cases}
m(\flat({Bundle_1}), x) & \text{when } x \in \flat({Bundle_2})\\
0 & \text{otherwise} \enspace.
\end{cases}
\]

\subsection{Resource contribution games}

We formally define our models of resource contribution games.
\begin{definition}
A \emph{resource contribution game} (RCG) is a tuple $G = (N,\gamma_1, \ldots, \gamma_n, \epsilon_1, \ldots, \epsilon_n)$ where:
\begin{itemize}[leftmargin=*]
\item $N = \{1, \ldots, n\}$ is a finite set of players;
\item $\gamma_i$ is a resource bundle ($i$'s goal, or objective);
\item $\epsilon_i$ is a finite resource bag ($i$'s endowment).
\end{itemize}
\end{definition}
We define the set of possible actions, or choices, of $i$ as the set of multisets (resource bags) $\ch_i(G) = \{C \mid C \subseteq \epsilon_i \}$, and the set of \emph{profiles} in $G$ as
$\ch(G) = \prod_{i \in N} \ch_i(G)$.
\begin{example}
  Say $A$ and $B$ are atomic resources, we can represent the complex resource composed of one $A$ and one $B$ as the resource bundle $A \comb B$.
  We can use $\gamma_i = A \comb B$ to express that player~$i$ wants $A$ and $B$ simultaneously.
  The endowment $\epsilon_i = \{B, B, A \comb B\}$ is different from the endowment $\{A, B, B, B\}$. With the former, player~$i$ cannot contribute an $A$ without also contributing a $B$.
\end{example}
When $P= (C_1, \ldots , C_n) \in \ch(G)$ and $1 \leq i \leq k$, then $P_i = C_i$, and $P_{-i} = (C_1, \ldots , C_{i-1}, C_{i+1}, \ldots, C_n)$. That is, $P_{-i}$ denotes $P$ without player $i$'s contribution $P_i$. When $P'_i \in \ch_i(G)$ and $P \in \ch(G)$, we note $(P_{-i}, P'_i)$ the profile obtained by substituting the choice of player~$i$ in the profile $P$ with his choice $P'_i$. The \emph{outcome} of a profile $P = (C_1, \ldots , C_n)$ is given by the resource bag \[\out(P) = \biguplus_{1\leq i \leq n} C_i \enspace .\]

We show how our examples in the introduction can be modelled as RCGs, and we illustrate the notions introduced so far.
\begin{example}[continues=ex:bringyourownfood]
A glass of wine is represented by $\var{w}$, one portion of rice is represented by $\var{r}$, and an onion is represented by $\var{o}$.
The bottle of wine is represented as 
the bundle $\bundle{\var{w}, \var{w}, \var{w}, \var{w}}$ of four glasses of wine. A portion of risotto is represented as the bundle $\bundle{\var{w}, \var{r}, \var{o}}$ of a glass of wine, one portion of rice, and one onion.

You are represented by player~$y$, Bruno is represented by player~$b$, and Carmen is represented by player~$c$.

As regards the first part of the example, with you, Bruno, and Carmen, the game is formally defined as:
$\Gexone{.1} = (\{y, b, c\},\gamma_y = \bundle{\var{w}, \var{w}}, \gamma_b = \bundle{\var{w}, \var{r}, \var{o}}, \gamma_c = \bundle{\var{w}, \var{r}, \var{o}}, \epsilon_y = \{\bundle{\var{w}, \var{w}, \var{w}, \var{w}} \}, \epsilon_b = \{\var{r}, \var{r}, \var{r}\}, \epsilon_c = \{\var{o}, \var{o}\})$.

If Edward joins the party, as in the second part of the example, then
his objective is $\gamma_e = \bundle{\var{w}, \var{w}, \var{w}}$
and his endowment is $\epsilon_e = \emptyset$.
We obtain the RCG $\Gexone{.2}$.
\end{example}

\begin{example}[continues=ex:telecommunication]
The symbols $\var{3G}$ and $\var{4G}$ are atomic resource types,
and $a$ and $b$ are the two telecom companies.

The scenario of Example~\ref{ex:telecommunication} can be formalized
as the RCG $\Gextwo = (\{a,b\},\gamma_a, \gamma_b, \epsilon_a, \epsilon_b)$, where $\gamma_a = \bundle{\var{3G}, \var{3G}, \var{4G}}$ and $\gamma_b = \bundle{\var{3G}, \var{3G}, \var{4G}, \var{4G}}$ are their respective goals, and $\epsilon_a = \{\var{3G},\var{3G},\var{3G}\comb \var{3G}\}$ and $\epsilon_b = \{\var{3G}, \var{4G}, \var{4G}\comb \var{4G}\}$ are their respective endowments.

The choices of player~$a$ are $\emptyset$, $\{\var{3G} \}$, $\{\var{3G}\comb \var{3G} \}$, $\{\var{3G}, \var{3G} \}$, $\{\var{3G}, \var{3G}\comb \var{3G} \}$, and $\{\var{3G}, \var{3G}, \var{3G}\comb \var{3G} \}$.

The multiset $\{ \var{3G}, \var{4G} \}$ is one of the choices of player~$b$. The tuple $P = (\{\var{3G}\comb \var{3G} \}, \{ \var{3G}, \var{4G} \})$ is a profile of $G$. The outcome of $P$ is the bag $\out(P) = \{ \var{3G}\comb\var{3G}, \var{3G}, \var{4G}\}$. Its decomposition into a bag of atomic resources is $\flat^\bullet(\out(P)) = \{ \var{3G}, \var{3G}, \var{3G}, \var{4G}\}$.





\end{example}

\begin{remark}
In RCGs, players are forced to contribute resources before using them. In the terminology of public good games~\cite{MARWELL1981295}, we can say that there is no \emph{private exchange}, and only a \emph{group exchange}.
This is only a limitation of RCGs as games for experimental economics. Just like in public good games the participants have no reason to invest on the group exchange at all, in RCGs with a private exchange, the players would have no reason to contribute resources useful to them on the group exchange.
If the modeller prefers to have reserved resources that a player can use privately towards his goal, it suffices to pre-process the game as follows. If a player wants $A \comb R_1 \comb \ldots \comb R_k$ and is endowed with $\bag{A, R'_1 \ldots, R'_m}$, remove an occurrence of $A$ from his goal and an occurrence from his endowment.
\end{remark}

\subsection{Player satisfaction and resource contention}
Let $G$ be the RCG $(N,\gamma_1, \ldots, \gamma_n, \epsilon_1, \ldots, \epsilon_n)$, and let $P$ be a profile in $\ch(G)$. We say player~$i$ is \emph{potentially satisfied} by the profile $P$ when the resources available in $P$ can be transformed into the goal $\gamma_i$, possibly with some left-over resources, that is $\out(P) \transform \gamma_i$.
For every RCG $G$ and player~$i$, we define $\ISAT(G,i)$ to be the set of profiles in $G$ by which player~$i$ is potentially satisfied. Formally:
\[
\ISAT(G,i) = \{ P \in \ch(G) \mid \out(P) \transform \gamma_i \} \enspace .
\]
For convenience, we also write $\SATP(G,P)$ to denote the set of players potentially satisfied in the profile $P \in \ch(G)$ of the game $G$, that is:
\[
\SATP(G,P) = \{ i \in N \mid \out(P) \transform \gamma_i \} \enspace .
\]

Player~$i$ is potentially satisfied in a profile $P$ when the resources in $\out(P)$ can be transformed into the resource bundle $\gamma_i$. But it might be that some resources are oversubscribed.
This yields resource contention.
The profile $P$ is contentious when the resource bag $\out(P)$ cannot be transformed into the resource bundle made up of all the objectives of the otherwise potentially satisfied players.
We write $\CONT(G)$ to denote the set of contentious profiles in $G$. Formally:
\[
\CONT(G) = \left\{ P \in \ch(G) \;\middle|\; \out(P) \not\transform \gamma_{s_1} \comb \ldots \comb \gamma_{s_k} \text{ where } \{s_1, \ldots, s_k\} = \SATP(G,P) \right\} \enspace .
\]

\begin{example}[continues=ex:telecommunication]
In the RCG $\Gextwo$ of Example~\ref{ex:telecommunication},
the profile 
$(\{\var{3G}, \var{3G}\}, \{\var{4G} \comb \var{4G}\})$ is in $\ISAT(\Gextwo,a)$ and $\ISAT(\Gextwo,b)$, and it is also in $\CONT(\Gextwo)$. Indeed, both players are potentially satisfied in $(\{\var{3G}, \var{3G}\}, \{\var{4G} \comb \var{4G}\})$, but the resources in the profile are then oversubscribed.
\end{example}

A distinction can now be made between \emph{public} and \emph{private} resource contention.
That is, the players might be averse to resource contentions that affect them, but might be indifferent to resource contentions about resources they personally have no claim about.
\begin{example}[label=ex:contention]
Consider the RCG $\Gexcont$ with $N = \{1,2,3\}$, $\epsilon_1 = \{A\comb B\}$, $\gamma_1 = A$, $\epsilon_2 = \epsilon_3 = \emptyset$, and $\gamma_2 = \gamma_3 = B$.
Although the profile $(\{A \comb B\}, \emptyset, \emptyset)$ is contentious, player~$1$ should not have to worry about it, because the contention is about the resource $B$ being disputed between player~$2$ and player~$3$. Player~$1$ should find a profile containing $A$ perfectly good.
\end{example}
We thus define $\CONT(G,i)$, the set of profiles that are contentious for player~$i$.
\begin{multline*}
\CONT(G,i) = \left\{ P \in \ch(G) \;\middle|\; P \in \ISAT(G,i) \text{ and } \out(P) \not\transform
{\gamma_{s_1}}_{\mid \gamma_i} \comb \ldots \comb {\gamma_{s_k}}_{\mid \gamma_i}\right.
\\\left. \text{ where } \{s_1, \ldots, s_k\} = \SATP(G,P) \right\} \enspace .
\end{multline*}
\begin{example}[continues=ex:contention]
The profile $(\{A \comb B\}, \emptyset, \emptyset)$ is in $\ISAT(\Gexcont,1)$, $\ISAT(\Gexcont, 2)$, and\linebreak $\ISAT(\Gexcont, 3)$. It is contentious: in $\CONT(\Gexcont)$, because $\out(P)\not\transform A \comb B \comb B$.
But it is not contentious for player~$1$: not in $\CONT(\Gexcont, 1)$, because
${\gamma_1}_{\mid \gamma_1} \comb {\gamma_2}_{\mid \gamma_1} \comb {\gamma_3}_{\mid \gamma_1} = A_{\mid A} \comb B_{\mid A} \comb B_{\mid A} = A \comb \emptybundle \comb \emptybundle$, and $\out(P)\transform A$.
\end{example}

To help with intuition, we already provide a few basic properties about resource contention.
\begin{fact}\label{fact:basic-properties}
  The following statements hold:
  \begin{enumerate}[leftmargin=*]
  \item \label{fact:trivial-uncontentious} The profile $(\emptyset, \ldots, \emptyset)$ is never contentious.
  \item \label{fact:cont-ind-cont-union} $\CONT(G) = 
    \bigcup_{i \in N} \CONT(G,i)$.
  \item \label{fact:private-cont-at-least-two} If $P \in \CONT(G,i)$ then $\exists j \not = i$, $P \in \CONT(G,j)$.
  \end{enumerate}
\end{fact}
The profile whose outcome is the empty set of resources is not contentious. 
The set of contentious profiles is the union of the profiles that are contentious for a player. If a profile is contentious for a player, then there is another distinct player for whom the profile is also contentious.

Players are surplus-tolerant. If a player is potentially satisfied in a profile, it will also be potentially satisfied in a profile resulting from any player adding more resources.
\begin{fact}\label{fact:surplus-tolerant-monotony}
If $(P_{-j}, P_j) \in \ISAT(G,i)$ and $P_j \subseteq Q_j$ then $(P_{-j}, Q_j) \in \ISAT(G,i)$.
\end{fact}
In contrast, resource contention is a non-monotonic property of a profile. From a contentious profile one could add resources and obtain a non-contentious profile. Also, from a contentious profile, one could remove resources and obtain a non-contentious profile. The following example illustrates that.
\begin{example}
Let $G = (N, \gamma_1, \gamma_2, \epsilon_1, \epsilon_2)$  be the RCG with $N = \{1,2\}$, $\gamma_1 = \gamma_2 = A$, $\epsilon_1 = \{A,A\}$, and $\epsilon_2 = \emptyset$. It is the case that $(\{A\},\emptyset)$ is a contentious profile. However, $(\emptyset,\emptyset)$, and $(\{A,A\},\emptyset)$ are not. We have $\CONT(G) = \{(\{A\},\emptyset)\}$.
\end{example}
\takeaway{ 
This will be a difficulty for finding profitable deviations, and deciding whether a profile is a Nash equilibrium in presence of contention-averse preferences.
} 

\subsection{Preferences}

We will consider three \emph{kinds of preference}: contention-tolerant preferences ($\pct$), (public) contention-averse preferences ($\pca$), private contention-averse preferences ($\ppca$).
To each kind of preferences corresponds a notion of what constitutes a `good' profile for a player.
%
\begin{definition}\label{def:good}
  A profile $P$ is \emph{good} for player~$i$ according to $\pref$ if:
  \begin{description}
  \item[($\pref = \pct$)] $P \in \ISAT(G,i)$.
  \item[($\pref = \pca$)] $P \in \ISAT(G,i)$ and $P \not \in \CONT(G)$.
  \item[($\pref = \ppca$)] $P \in \ISAT(G,i)$ and $P \not \in \CONT(G,i)$.
  \end{description}
\end{definition}
That is, with contention-tolerant preferences, a profile is good for a player whenever the profile potentially satisfies him. With (public) contention-averse preferences, a profile is good for a player when the profile potentially satisfies him, and the profile is not contentious. With private contention-averse preferences, a profile is good for a player when the profile potentially satisfies him, and the profile is not contentious for him.
We note $\GOOD^{\pref}(G,i)$ the set of profiles in $G$ that are good for player~$i$ according to $\pref \in \{\pct, \pca, \ppca\}$.
The following fact is a straightforward consequence of the definitions.
\begin{fact}\label{fact:pcasuppcabsubpct}
  Given a game $G$ and a player~$i$, we have $\GOOD^\pca(G,i) \subseteq \GOOD^\ppca(G,i)\linebreak \subseteq \GOOD^\pct(G,i)$.
\end{fact}

We can now define the three kinds of single-minded parsimonious preferences over profiles used in this paper.
\begin{definition}\label{def:preferences}
In an RCG $G = (N,\gamma_1, \ldots, \gamma_n,
\epsilon_1, \ldots, \epsilon_n)$, we say that player $i \in N$
\emph{strictly prefers} $P \in \ch(G)$ over $Q \in \ch(G)$ according to $\pref$ (noted $Q \prec^\pref_i P$) iff one of the following conditions is
satisfied:
\begin{enumerate}[leftmargin=*]
\item $P \not \in \GOOD^\pref(G,i)$, $Q \not \in \GOOD^\pref(G,i)$ and $P_i \subset Q_i$;
\item $P \in \GOOD^\pref(G,i)$, and $Q \not \in \GOOD^\pref(G,i)$;
\item $P \in \GOOD^\pref(G,i)$, $Q \in \GOOD^\pref(G,i)$ and $P_i \subset Q_i$.
\end{enumerate}
\end{definition}

\takeaway{ 
From, Fact~\ref{fact:surplus-tolerant-monotony} players are surplus tolerant, and their attitude towards resource contention depends on $\pref$ and Definition~\ref{def:good}.
}

\takeaway{
\begin{remark}
  RCGs in presence of resource-tolerant preferences correspond to
  the RCGs from~\cite{Tr16ijcai}, in presence of parsimonious preferences,
  and  instantiated with the fragment of Linear Logic corresponding to the language $A ::= \mathbf{1} \mid p \mid A \otimes A$, with $p$ an atomic variable.
  %
See details in Sec.~\ref{sec:conclusions}.
\end{remark}
} 

\section{Nash equilibria and complexity}
\label{sec:decision-prob-complexity}

In this section, we define pure Nash equilibria in RCGs.
We show that Nash equilibria may not exist, and that 
the sets of Nash equilibria may be distinct
for every kind of attitude towards contention.
We show that deciding whether a profile is a Nash equilibrium is \coNP-complete in presence of public and private contention-averse preferences and in \PTIME in presence of contention-tolerant preferences.

\medskip

It is customary to define Nash equilibria in terms of profitable deviations, which we introduce now.
\begin{definition}
Let $G = (N,\gamma_1, \ldots, \gamma_n, \epsilon_1, \ldots, \epsilon_n)$ be an RCG. We say that $P'_i \in \ch_i(G)$ is a \emph{profitable deviation}, according to $\pref$, from the profile $P \in \ch(G)$ for player~$i$ when player~$i$ strictly prefers $(P_{-i}, P'_i)$ over $P$, according to $\pref$.
\end{definition}

A Nash equilibrium is a profile in which no player has a profitable deviation.
\begin{definition}Let $G = (N,\gamma_1, \ldots, \gamma_n, \epsilon_1, \ldots, \epsilon_n)$ be an RCG. A profile $P \in \ch(G)$ is a (pure) \emph{Nash equilibrium} according to $\pref$ iff there is no profitable deviation according to $\pref$ from $P$ for any player in $N$.
\end{definition}
Let us note $\NE^\pref(G)$ the set of profiles in $\ch(G)$ which are Nash equilibria of $G$ according to $\pref$.
We can finally illustrate this last definition with our running examples.
\begin{example}[continues=ex:bringyourownfood]
  If Edward does not join the party, recall that the game is formalized as $\Gexone{.1} = (\{y, b, c\},\gamma_y = \bundle{\var{w}, \var{w}}, \gamma_b = \bundle{\var{w}, \var{r}, \var{o}}, \gamma_c = \bundle{\var{w}, \var{r}, \var{o}}, \epsilon_y = \{\bundle{\var{w}, \var{w}, \var{w}, \var{w}} \}, \epsilon_b = \{\var{r}, \var{r}, \var{r}\}, \epsilon_c = \{\var{o}, \var{o}\})$.
  We have: 
in the case of contention-tolerance
$\NE^\pct(\Gexone{.1}) = \{
(\bundle{\var{w}, \var{w}, \var{w}, \var{w}}, \emptyset, \emptyset),
(\bundle{\var{w}, \var{w}, \var{w}, \var{w}}, \{\var{r}\}, \{\var{o}\})
\}$;
in both cases of contention-aversity $\NE^\pca(\Gexone{.1}) = \NE^\ppca(\Gexone{.1}) = \{
(\bundle{\var{w}, \var{w}, \var{w}, \var{w}}, \emptyset, \emptyset),
(\bundle{\var{w}, \var{w}, \var{w}, \var{w}}, \{\var{r}, \var{r}\}, \{\var{o}, \var{o}\})
\}$.

\medskip

Let us see why $(\bundle{\var{w}, \var{w}, \var{w}, \var{w}}, \{\var{r}\}, \{\var{o}\})$ is not a Nash equilibrium in $\Gexone{.1}$ for contention-averse preferences.

In which profiles are the players potentially satisfied?
\begin{itemize}
\item $\ISAT(\Gexone{.1}, y)$ contains all the profiles where the choice of $y$ is $\{\bundle{\var{w}, \var{w}, \var{w}, \var{w}}\}$.
\item $\ISAT(\Gexone{.1}, b)$ contains all the profiles where the choice of $y$ is $\{\bundle{\var{w}, \var{w}, \var{w}, \var{w}}\}$, the choice of $b$ is at least $\{r\}$, and the choice of $c$ is at least $\{o\}$.
\item $\ISAT(\Gexone{.1}, c) = \ISAT(\Gexone{.2}, b)$. 
\end{itemize}
So all the players are potentially satisfied in $(\bundle{\var{w}, \var{w}, \var{w}, \var{w}}, \{\var{r}\}, \{\var{o}\})$.

As an aside, let us first observe that the profile not being a Nash equilibrium is not because of you.
Indeed, the profile is not contentious for you.
\begin{itemize}
  \item ${\gamma_y}_{\mid \gamma_y} = (\var{w} \comb \var{w})_{\mid \var{w} \comb \var{w}} = \var{w} \comb \var{w}$,
  \item ${\gamma_b}_{\mid \gamma_y} = (\var{w} \comb \var{r} \comb \var{o})_{\mid \var{w} \comb \var{w}} = \var{w} $,
  \item ${\gamma_c}_{\mid \gamma_y} = (\var{w} \comb \var{r} \comb \var{o})_{\mid \var{w} \comb \var{w}} = \var{w} $.
\end{itemize}
Clearly, one can transform the outcome of $(\bundle{\var{w}, \var{w}, \var{w}, \var{w}}, \{\var{r}\}, \{\var{o}\})$ into a bundle with four glasses of wine $\var{w}$. So $\CONT(\Gexone{.1}, y)$ does not contain the profile (which is then contained in $\GOOD^\ppca(\Gexone{.1}, y)$). We can also see that if you were not contributing the bottle of wine, you would not be potentially satisfied anymore. In other words, you have no profitable deviation from the profile $(\bundle{\var{w}, \var{w}, \var{w}, \var{w}}, \{\var{r}\}, \{\var{o}\})$.

But the profile is contentious for Bob and Carmen:
\begin{itemize}
  \item ${\gamma_y}_{\mid \gamma_b} = (\var{w} \comb \var{w})_{\mid \var{w} \comb \var{r} \comb \var{o}} = \var{w} \comb \var{w}$,
  \item ${\gamma_b}_{\mid \gamma_b} = (\var{w} \comb \var{r} \comb \var{o})_{\mid \var{w} \comb \var{r} \comb \var{o}} = \var{w} \comb \var{r} \comb \var{o}$,
  \item ${\gamma_c}_{\mid \gamma_b} = (\var{w} \comb \var{r} \comb \var{o})_{\mid \var{w} \comb \var{r} \comb \var{o}} = \var{w} \comb \var{r} \comb \var{o}$.
\end{itemize}
Clearly, it is not possible to transform the outcome of the profile into a bundle with two portions of rice (and two onions). So, $(\bundle{\var{w}, \var{w}, \var{w}, \var{w}}, \{\var{r}\}, \{\var{o}\})$ is in $\CONT(\Gexone{.1}, b)$, and so also in $\CONT(\Gexone{.1})$. It means that it is in neither $\GOOD^\ppca(\Gexone{.1}, b)$ or $\GOOD^\pca(\Gexone{.1}, b)$. The exact same is true for agent~$c$.

Finally, $(\bundle{\var{w}, \var{w}, \var{w}, \var{w}}, \{\var{r}\}, \{\var{o}\})$ is not a Nash equilibrium when the preferences are contention-averse, because it is not good for either Bob or Carmen, and both can profitably deviate to $\emptyset$. Indeed, it would result in another profile that is not good for them either, but their contribution would be strictly included in their choice in $(\bundle{\var{w}, \var{w}, \var{w}, \var{w}}, \{\var{r}\}, \{\var{o}\})$. That is, it is the case that:
\begin{itemize}
\item $(\bundle{\var{w}, \var{w}, \var{w}, \var{w}}, \{\var{r}\}, \{\var{o}\}) \prec^\pref_b (\bundle{\var{w}, \var{w}, \var{w}, \var{w}}, \emptyset, \{\var{o}\})$,
  \item $(\bundle{\var{w}, \var{w}, \var{w}, \var{w}}, \{\var{r}\}, \{\var{o}\}) \prec^\pref_c (\bundle{\var{w}, \var{w}, \var{w}, \var{w}}, \{\var{r}\}, \emptyset)$.
\end{itemize}

\medskip

If Edward joins the party, recall that the game is formalized as
 $\Gexone{.2} = ( \{y, b, c, e\},\gamma_y = \bundle{\var{w}, \var{w}}, \gamma_b = \bundle{\var{w}, \var{r}, \var{o}}, \gamma_c = \bundle{\var{w}, \var{r}, \var{o}}, \gamma_e = \bundle{\var{w}, \var{w}, \var{w}}, \epsilon_y = \{\bundle{\var{w}, \var{w}, \var{w}, \var{w}} \}, \epsilon_b = \{\var{r}, \var{r}, \var{r}\}, \epsilon_c = \{\var{o}, \var{o}\}, \epsilon_e = \emptyset )$.
In the case of contention-tolerance
$\NE^\pct(\Gexone{.2}) = \{
(\bundle{\var{w}, \var{w}, \var{w}, \var{w}}, \emptyset, \emptyset, \emptyset),
(\bundle{\var{w}, \var{w}, \var{w}, \var{w}}, \{\var{r}\}, \{\var{o}\}, \emptyset)
\}$;
in both cases of contention-aversity $\NE^\pca(\Gexone{.2}) = \NE^\ppca(\Gexone{.2}) =
\{
(\emptyset, \emptyset, \emptyset, \emptyset)
\}$.
\end{example}

\begin{example}[continues=ex:telecommunication]
Consider again $\Gextwo = (\{a,b\},\gamma_a = \bundle{\var{3G}, \var{3G}, \var{4G}}, \gamma_b = \bundle{\var{3G}, \var{3G}, \var{4G}, \var{4G}}, \epsilon_a = \{\var{3G},\var{3G},\linebreak \var{3G}\comb \var{3G}\}, \epsilon_b = \{\var{3G}, \var{4G}, \var{4G}\comb \var{4G}\})$.
We have:
in the case of contention-tolerance
$\NE^\pct(\Gextwo) = \{(\emptyset, \emptyset), (\{\var{3G}, \var{3G}\},\linebreak \{\var{4G}\comb \var{4G}\}), (\{\var{3G}\}, \{\var{3G}, \var{4G}\comb \var{4G}\}), (\{\var{3G}, \var{3G}\}, \{\var{4G}\comb \var{4G}\})\}$;
in both cases of contention-aversity $\NE^\pca(\Gextwo) = \linebreak \NE^\ppca(\Gextwo) = \{(\emptyset, \emptyset), (\{\var{3G}, \var{3G} \comb \var{3G}\}, \{\var{3G}, \var{4G}, \var{4G}\comb \var{4G}\}), (\{\var{3G}, \var{3G}, \var{3G} \comb \var{3G}\}, \{\var{4G}, \var{4G}\comb \var{4G}\})\}$.\footnote{From Fact~\ref{fact:basic-properties}.\ref{fact:private-cont-at-least-two}, when there are only two players, contention and private contention coincide, and so do the sets of Nash equilibria.}
\end{example}

\takeaway{
\begin{proposition}\label{prop:pca-nash-non-contentious}
  If $P \in \NE^\pca(G)$ then $P \not \in \CONT(G)$.
\end{proposition}
\begin{proof}
If $P \in \CONT(G)$, then $P$ is not the empty profile (Fact~\ref{fact:basic-properties}.\ref{fact:trivial-uncontentious}). Hence, there is a player~$i$ such that $P_i \not = \emptyset$. If $\pref = \pca$, since $P$ is contentious, $P \not \in \GOOD^\pref(G,i)$, and player~$i$ can parsimoniously profitably deviate to $\emptyset$. So $P$ is not a Nash equilibrium.
\end{proof}

\begin{proposition}\label{prop:ppca-nash-empty-in-contentious}
  If $P \in \NE^\ppca(G)$ and $P_i \not = \emptyset$, then $P \not \in \CONT(G, i)$.
\end{proposition}
\begin{proof}
If $P \in \CONT(G, i)$, and $P_i \not = \emptyset$, then $P \not\in \GOOD^\ppca(G,i)$ and player~$i$ has a parsimonious profitable deviation to $\emptyset$. So $P$ is not a Nash equilibrium.
\end{proof}
}

For every kind of preferences, there are games that do not admit any Nash equilibrium. In fact we prove something stronger: there is a game that does not admit a Nash equilibrium for any kind of preferences.
\begin{proposition}\label{prop:no-nash}
  There is a game $G$ such that $\NE^\pref(G) = \emptyset$
  for every kind of preferences\linebreak $\pref \in \{\pct, \pca, \ppca\}$.
\end{proposition}
\begin{proof}
  Consider the game $G$, with two players~$1$ and $2$, where $\epsilon_1 = \epsilon_2 = \{A \comb B\}$, $\gamma_1 = A$, and $\gamma_2 = B \comb B$. An illustration of the four profiles and their outcomes follows.

  \begin{center}
    \begin{tabular}{lcc}
      \toprule
      ~           & $\emptyset$ & $\{A \comb B\}$\\
      \midrule
      $\emptyset$ & $\emptyset$ &$\{A \comb B\}$ \\
      $\{A \comb B\}$ & $\{A \comb B\}$  & $\{A \comb B, A\comb B\}$ \\
      \bottomrule
    \end{tabular}
  \end{center}

  We have:
  \begin{itemize}
  \item $\ISAT(G, 1) = \{ (\emptyset, \{A \comb B\}), (\{A \comb B\}, \emptyset), (\{A \comb B\}, \{A \comb B\})\}$,
  \item $\ISAT(G, 2) = \{(\{A \comb B\}, \{A \comb B\})\}$.
  \end{itemize}
  By definition we have:
  \begin{itemize}
  \item $\GOOD^\pct(G,1) = \ISAT(G,1)$ and $\GOOD^\pct(G,2) = \ISAT(G,2)$.
  \end{itemize}
  But it is also the case that:
  \begin{itemize}
  \item $\CONT(G,1) = \emptyset$,
  \item $\CONT(G,2) = \emptyset$.
  \end{itemize}
  Which means that $\CONT(G) = \emptyset$.
  Hence, we also have:
  \begin{itemize}
  \item $\GOOD^\pca(G,1) = \ISAT(G,1)$, $\GOOD^\pca(G,2) = \ISAT(G,2)$,
  \item $\GOOD^\ppca(G,1) = \ISAT(G,1)$ and $\GOOD^\ppca(G,2) = \ISAT(G,2)$.
  \end{itemize}
  It means that in $G$, the preferences $\pct$, $\pca$, and $\ppca$ coincide. We argue that there is a profitable deviation in every profile of $G$. 
  \begin{itemize}
  \item From profile $(\emptyset, \emptyset)$, player~$1$ prefers $(\{A\comb B\}, \emptyset)$, because $(\emptyset, \emptyset)\not\in \GOOD^\pref(G,1)$ and $(\{A\comb B\}, \emptyset) \in \GOOD^\pref(G,1)$.
  \item From profile $(\emptyset, \{A\comb B\})$, player~$2$ prefers $(\emptyset, \emptyset)$ because $(\emptyset, \{A\comb B\}) \not \in \GOOD^\pref(G,2)$, $(\emptyset, \emptyset)  \not \in \GOOD^\pref(G,2)$, and $\emptyset \subset \{A\comb B\}$.
  \item From profile $(\{A\comb B\}, \emptyset)$, player~$2$ prefers $(\{A\comb B\}, \{A\comb B\})$, because $(\{A\comb B\}, \emptyset) \not \in \GOOD^\pref(G,2)$ and $(\{A\comb B\}, \{A\comb B\}) \in  \GOOD^\pref(G,2)$.
  \item From profile $(\{A\comb B\}, \{A\comb B\})$, player~$1$ prefers $(\emptyset, \{A\comb B\})$, because $(\{A\comb B\}, \{A\comb B\}) \in\linebreak \GOOD^\pref(G, 1)$,  $(\emptyset, \{A\comb B\}) \in \GOOD^\pref(G, 1)$, and $\emptyset \subset \{A\comb B\}$.
  \end{itemize}

\end{proof}

The sets of Nash equilibria wrt.\ two different kinds of preference can be pairwise distinct and even non-nested.
\begin{proposition}
\label{prop:all-NE-distinct}
  For every pair of distinct kinds of preference $\pref, \pref' \in \{\pct, \pca, \ppca\}$, there is a game $G$ such that $\NE^{\pref}(G) \not \subseteq \NE^{\pref'}(G)$.
  \end{proposition}
  \begin{proof}
    For the cases where $\pref$ or $\pref'$ is $\pct$, c.f.~Example~\ref{ex:bringyourownfood}.
    The profile $(\bundle{\var{w}, \var{w}, \var{w}, \var{w}}, \{\var{r}\}, \{\var{o}\})$ is in $\NE^\pct(\Gexone{.1})$ but in neither $\NE^\pca(\Gexone{.1})$ or $\NE^\ppca(\Gexone{.1})$.
The profile $(\bundle{\var{w}, \var{w}, \var{w}, \var{w}}, \{\var{r}, \var{r}\}, \{\var{o}, \var{o}\})$ is in both $\NE^\pca(\Gexone{.1})$ and $\NE^\ppca(\Gexone{.1})$, but it is not in $\NE^\pct(\Gexone{.1})$.

For the two cases where $\pref = \pca$ and $\pref' = \ppca$, and where $\pref = \ppca$ and $\pref' = \pca$, consider $G$, where $\gamma_1 = A$, $\gamma_2 = \gamma_3 = B$, $\epsilon_1 = \{A\comb B, B\}$, $\epsilon_2 = \epsilon_3 = \emptyset$.
    We have $\NE^\ppca(G) = \{ (\{A\comb B\}, \emptyset, \emptyset) \}$, and
    $\NE^\pca(G) = \{ (\{A\comb B, B\}, \emptyset, \emptyset) \}$. We provide the details.

    Player~$1$ is potentially satisfied when he contributes at least $A \comb B$, while player~$2$ and player~$3$ are potentially satisfied when player~$1$ contributes more than nothing.
  \begin{itemize}
  \item $\ISAT(G, 1) = \{ (\{A \comb B\}, \emptyset, \emptyset), (\{A \comb B, B\}, \emptyset, \emptyset) \}$.
  \item $\ISAT(G, 2) = \ISAT(G, 3) = \{(\{B\}, \emptyset, \emptyset), (\{A \comb B\}, \emptyset, \emptyset), (\{A \comb B, B\}, \emptyset, \emptyset)\}$.
  \end{itemize}
  We have ${\gamma_1}_{\mid \gamma_2} = {\gamma_1}_{\mid \gamma_3} = \emptybundle$, and ${\gamma_2}_{\mid \gamma_1} = {\gamma_3}_{\mid \gamma_1} = \emptybundle$, and 
  ${\gamma_2}_{\mid \gamma_3} = {\gamma_3}_{\mid \gamma_2} = B$.
  Hence:
  \begin{itemize}
  \item $\CONT(G, 1) = \emptyset$.
  \item $\CONT(G, 2) = \CONT(G, 3) = \{(\{A \comb B\}, \emptyset, \emptyset)\}$. (Because all players are potentially satisfied in $(\{A \comb B\}, \emptyset, \emptyset)$, but $\out((\{A \comb B\}, \emptyset, \emptyset)) \not\transform {\gamma_1}_{\mid \gamma_2} \comb {\gamma_2}_{\mid \gamma_2}  \comb {\gamma_3}_{\mid \gamma_2}$ and $\out((\{A \comb B\}, \emptyset, \emptyset)) \not\transform {\gamma_1}_{\mid \gamma_3} \comb {\gamma_2}_{\mid \gamma_3}  \comb {\gamma_3}_{\mid \gamma_3}$.)
  \end{itemize}
  Which also means that $\CONT(G) = \{(\{A \comb B\}, \emptyset, \emptyset)\}$.
  We thus have:
  \begin{itemize}
    \item $\GOOD^\pca(G, 1) = \{ (\{A \comb B, B\}, \emptyset, \emptyset) \}$,
    \item $\GOOD^\ppca(G, 1) = \{ (\{A \comb B\}, \emptyset, \emptyset), (\{A \comb B, B\}, \emptyset, \emptyset) \}$.
  \end{itemize}

  Since player~$2$ and player~$3$ can only choose $\emptyset$, they cannot have profitable deviations.
  So $\NE^\pca(G)$ only contains $(\{A \comb B, B\}, \emptyset, \emptyset)$, as it is the only good profile for player~$1$. On the other hand, with private contention-averse preferences,
$(\{A \comb B, B\}, \emptyset, \emptyset) \prec^\ppca_1 (\{A \comb B\}, \emptyset, \emptyset)$ (because both profiles are good for player~$1$, but he contributes strictly less in $(\{A \comb B\}, \emptyset, \emptyset)$), so $\NE^\ppca(G)$ only contains $(\{A \comb B\}, \emptyset, \emptyset)$.
  \end{proof}

We will investigate the decision problem related to the membership in the set $\NE^\pref(G)$, for $\pref \in \{\pct, \pca, \ppca\}$.

\noindent\begin{minipage}{\linewidth}
\medskip
\hrule height 1pt
\smallskip
\textsf{NASH EQUILIBRIUM ($\NE^\pref$)}
\smallskip
\hrule
\begin{description}
\item[{\bf (in)}] A resource contribution game $G$ and a profile $P \in \ch(G)$.
\item[{\bf (out)}] $P \in \NE^\pref(G)$?
\end{description}
\hrule
\medskip
\end{minipage}

To start, we state the rudimentary results that are instrumental to study the complexity of our decision problems.
\begin{proposition}
\label{prop:ptime-satp-good}\label{prop:ptime-prefs}
For every kind of preferences $\pref \in \{\pct, \pca, \ppca\}$,
the problems of deciding whether $P \in \ISAT(G,i)$, $i \in \SATP(G,P)$, $P \in \CONT(G)$, $P \in \CONT(G,i)$, $P \in \GOOD^\pref(G,i)$, and $Q \prec^\pref_i P$, are all in \PTIME.
\end{proposition}
  \begin{proof}
    $P \in \ISAT(G,i)$ iff $\out(P) \transform \gamma_i$ iff $\flat(\gamma_i) \subseteq \flat^\bullet(\out(P))$. The functions $\flat$ and $\flat^\bullet$ respectively transform bundles and bags into multisets of atomic resources, of size linear in the size of the input. Checking that a multiset is included in another can be done efficiently.

    $i \in \SATP(G,P)$ is reducible to $P \in \ISAT(G,i)$.

    $P \in \CONT(G)$ iff (with $\SATP(G,P) = \{s_1, \ldots, s_k\}$) $\out(P) \not\transform {\gamma_{s_1}} \comb \ldots \comb {\gamma_{s_k}}$ iff
    $\flat({\gamma_{s_1}} \comb \ldots \comb {\gamma_{s_k}}) \not\subseteq \flat^\bullet(\out(P))$.
We have just seen that we can decide efficiently whether $i \in \SATP(G,P)$, so we can also compute $\SATP(G,P)$ efficiently. Now, it suffices to check whether a multiset is a subset of another, which can be done efficiently.

    $P \in \CONT(G,i)$ iff (with $\SATP(G,P) = \{s_1, \ldots, s_k\}$) $P \in \ISAT(G,i)$ and $\out(P) \not\transform  {\gamma_{s_1}}_{\mid \gamma_i} \comb \ldots \comb {\gamma_{s_k}}_{\mid \gamma_i}$. We can decide efficiently whether $i \in \SATP(G,P)$, so we can also compute $\SATP(G,P)$ efficiently. We can also efficiently compute the restriction of multisets ${\gamma_j}_{\mid \gamma_i}$. Now, it suffices to check whether a multiset is a subset of another, which can be done efficiently.
    
    No matter the kind of preferences $\pref$, from Definition~\ref{def:good} and the previous results, deciding\linebreak whether $P \in \GOOD^\pref(G,i)$ can also be done efficiently.

    So, for every profile $P'$ in an RCG $G$, and every player $i$, we can decide whether\linebreak $P' \in \GOOD^\pref(G,i)$ in polynomial time. Also, checking whether a multiset $P_i$ is a strict subset of another multiset $Q_i$ can be done efficiently. Hence, deciding whether $Q \prec^\pref_i P$ can be done in polynomial time.
\end{proof}

\subsection{Public and private contention-aversity}
\label{sec:case-contention-aversity}

As witnessed by Prop.~\ref{prop:all-NE-distinct}, Nash equilibria according to the preference kinds $\pca$ and $\ppca$ can be distinct, and even non-nested.
From the theoretical point of view of computational complexity, they are similar enough to be studied together. We can indeed solve them and show their hardness in a rather uniform manner.

\begin{proposition}\label{prop:ne-in-conp-pca-ppca}
  When $\pref \in \{\pca, \ppca\}$, the problem $\NE^\pref$ is in \coNP.
\end{proposition}
\begin{proof}
    To solve $P \not \in \NE^\pref(G)$, simply guess a pair $(i,P'_i) \in N \times \ch_i(G)$ and check whether $P \prec^\pref_i (P_{-i},P'_i)$, which by Prop.~\ref{prop:ptime-prefs} can be done in polynomial time.
  \end{proof}

%
To prove that $\NE^\pca$ and $\NE^\ppca$ are $\coNP$-hard, we first consider an auxiliary decision problem.
Given a resource contribution game $G$ and a player~$i$, the problem  $\EPD^\pref_\emptyset$ asks whether there is a profitable deviation according to $\pref$ from $(\emptyset, \ldots, \emptyset)$ for player~$i$.

\noindent\begin{minipage}{\linewidth}
\medskip
\hrule height 1pt
\smallskip
\textsf{PROFITABLE DEVIATION EXISTENCE FROM $(\emptyset, \ldots \emptyset)$ ($\EPD^\pref_\emptyset$)}
\smallskip
\hrule
\begin{description}
\item[{\bf (in)}] A resource contribution game $G$ and a player~$i$.
\item[{\bf (out)}] Is there a profitable deviation according to $\pref$ from $(\emptyset, \ldots, \emptyset)$ for player~$i$?
\end{description}
\hrule
\medskip
\end{minipage}

We show that $\EPD_\emptyset^\pca$ and $\EPD_\emptyset^\ppca$ are \NP-complete.
\begin{lemma}
\label{lemma:epd-empty-npc-pca-ppca}
  When $\pref \in \{\pca, \ppca\}$, the problem $\EPD_\emptyset^\pref$ is \NP-complete.
  \end{lemma}

  \begin{proof}
    In the following, let $\pref$ be either $\pca$ or $\ppca$.
    For membership, it is enough to non-deterministi\-cally guess a choice $C_i \in \ch_i(G)$ and check whether player~$i$ strictly prefers $(\emptyset, \ldots, C_i, \ldots \emptyset)$ over $(\emptyset, \ldots \emptyset)$ according to $\pref$, which can be done in polynomial time (Prop.~\ref{prop:ptime-prefs}).

    For hardness, we reduce \SAT. Let $\psi = c_1 \land \ldots \land c_m$ be a 3CNF over the set of propositional variables $\VAR(\psi)$, with $c_i = l_{i,1} \lor l_{i,2} \lor l_{i,3}$, where every literal $l_{i,j}$ is in $\VAR(\psi) \cup \{\lnot p \mid p \in \VAR(\psi)\}$.
    

    For every 3CNF $\psi$, we build the RCG $G^\psi = (\{1\} \cup N^{-}, \vec{\gamma}, \vec{\epsilon})$ over $\Res = \{c_i \mid 1 \leq i \leq m  \} \cup \{p, \lnot p \mid p \in \VAR(\psi)\}$ as follows:
    \begin{itemize}[leftmargin=*, wide, labelwidth=!, labelindent=0pt, itemindent=!]
      \item The objective of player~$1$ is $\gamma_1 = c_1 \comb \ldots \comb c_m$. That is, player~$1$ wants a resource bundle containing one resource representing every clause.
      The endowment of player~$1$ is $\epsilon_1 = \{c_i \comb l_{i,1}, c_i \comb l_{i,2}, c_i \comb l_{i,3} \mid 1 \leq i \leq m\}$. 
      For every clause, player~$1$ can contribute three or less of these resources. But every time he contributes one resource representing a clause resource, he cannot do without providing also one resource representing a literal in the clause.
    \item $N^{-} = \{2, \ldots, 1 + 3 \cdot \Card{\VAR(\psi)} \}$
    \item For every variable $p \in \VAR(\psi)$, the set $N^{-}$ contains 
 three players with objective $\gamma_{id(p,1)} = \gamma_{id(p,2)} = \gamma_{id(p,3)} = \gamma_1 \comb p \comb \lnot p$, and with empty endowments. (The function $id$ is an arbitrary bijection that takes one variable/integer couple in $\VAR(\psi) \times \{1, 2, 3\}$, and returns one player in $N^{-}$.)
      Every member of $N^{-}$ wants a resource representing a variable and a resource representing the negation of this same variable, and like player~$1$, also wants the resources corresponding to all clauses $c_i$. They are dummy players, having the empty set as their only action: $\epsilon_{id(p,1)} = \epsilon_{id(p,2)} = \epsilon_{id(p,3)} = \emptyset$.
    \end{itemize}
    
    We can show that $\EPD_\emptyset^\pref(G^\psi, 1)$ is true iff $\psi$ is satisfiable.

    We first sketch the idea of the construction and the proof.
    The main rationales behind the construction are that player~$1$ does not satisfy his objective in $(\emptyset, \ldots, \emptyset)$, and that player~$1$ must deviate to a non-contentious profile containing one $c_i \comb l_{i,j}$ for every clause $c_i$ in $\psi$.
We will argue that such a profitable deviation exists iff $\psi$ is satisfiable.
The presence in $N^{-}$ of players with objective $\gamma_1 \comb p \comb \lnot p$, for some $p \in \VAR(\psi)$, makes sure that a profile whose outcome contains all $c_i$'s but also a $p$ and a $\lnot p$ is contentious, and privately contentious to player~$1$. Thus, such a deviation from $(\emptyset, \ldots, \emptyset)$ cannot be profitable for player~$1$. Still, a deviation containing all $c_i$ but not both $p$ and $\lnot p$ for some $p \in \VAR(\psi)$ remains a profitable deviation.

    Right to left. Let $v$ be a valuation of $\VAR(\psi)$ and assume that $v(\psi) = true$.

    We construct a profile and we will show that it is a profitable deviation for player~$1$ from $(\emptyset, \ldots, \emptyset)$.
    Let $P_1 \in \ch_1(G^\psi)$ be the action $\{ l_{i,j} \comb c_i \mid 1 \leq i \leq m, v(l_{i,j}) = true\}$. Let $P \in \ch(G^\psi)$ be the profile $(P_1, \emptyset, \ldots, \emptyset)$.
    
    We claim that:
    \begin{enumerate}[leftmargin=*, wide, labelwidth=!, itemindent=!]
    \item \label{item:no-happy-dummy-player} \emph{No player $i \in N^-$ is potentially satisfied in $P$.}
    We have $v(\psi) = true$, so there is no $p \in \VAR(\psi)$ such that $\{p, \lnot p\} \subseteq \out(P)$. That is, $\out(P) \not \transform p \comb \lnot p$. This means that for every player $i \in N^-$, $P \not \in \ISAT(G^\psi,i)$.
    \item \emph{$(\emptyset, \ldots, \emptyset)$ is not a good profile for player~$1$.}
    $\biguplus_{1 \leq i \leq m}\{c_i\} \not \subseteq \out((\emptyset, \ldots, \emptyset)) = \emptyset$. So, $(\emptyset, \ldots, \emptyset) \not \in \ISAT(G^\psi,1)$, and so $(\emptyset, \ldots, \emptyset) \not \in \GOOD^\pref(G^\psi,1)$.
    \item \emph{$P$ is a good profile for player~$1$.}
    For every $1 \leq i \leq m$, we have $P_1 \transform c_i$, which is simply $\{c_i\} \subseteq P_1$. Consequently, we have $\biguplus_{1 \leq i \leq m}\{c_i\} \subseteq P_1$, which is exactly $P_1 \transform \gamma_1$. So $P \in \ISAT(G^\psi,1)$. Moreover, since no player $i \in N^-$ is potentially satisfied in $P$ (established in item~\ref{item:no-happy-dummy-player}), the profile $P$ is not contentious, and not contentious for player~$1$.
    We thus have $P \in \GOOD^\pref(G^\psi,1)$.
    \end{enumerate}
    We conclude that $P$ is a profitable deviation for player~$1$ from $(\emptyset, \ldots, \emptyset)$.


    Left to right. Let $P_1 \in \ch_1(G^\psi)$ be a profitable deviation from $(\emptyset, \ldots, \emptyset)$ for player~$1$.
We define $P = (P_1, \emptyset, \ldots, \emptyset)$.
    Being a profitable deviation for player~$1$, it means that $P \in \GOOD^\pref(G^\psi, 1)$.

    We define the valuation $v^P$ of $\psi$, such that for every $p\in \VAR(\psi)$, $v^P(p) = true$ iff $\out(P) \transform p$. We claim that:
    \begin{enumerate}[leftmargin=*, wide, labelwidth=!, itemindent=!]
    \item \emph{For every $i$, $1 \leq i \leq m$, there is $j \in \{1,2,3\}$ such that $c_i \comb l_{i,j} \in P_1$.} In other words, $P_1$ contains at least $m$ resources, and at least one for each clause. It is a direct consequence of $P \in \GOOD^\pref(G^\psi, 1)$, which implies $\out(P) \transform \gamma_1$.
    \item \emph{For every $i \in N^-$, we have $\out(P) \not \transform \gamma_i$.}
    Assume towards contradiction that there is $i \in N^-$ such that $\out(P) \transform \gamma_i$. That is, there is $p \in \VAR(\psi)$, such that  $\out(P) \transform \gamma_1 \comb p \comb \lnot p$. But it also means that there are two other players in $N^-$, also with objective $\gamma_1 \comb p \comb \lnot p$, which are potentially satisfied in profile $P$. Let $N^{-}(p)$ be this set of players. For every $i \in N^{-}(p)$, we have $P \in \ISAT(G^\psi,i)$, and ${\gamma_i}_{|\gamma_1} = c_1 \comb \ldots \comb c_m = \gamma_1$. Hence, together with player~$1$ there are four players who have a claim on a resource bundle $\gamma_1$.
    However, for each clause $c_i$, $1 \leq i \leq m$, there are at most three instances of the resource representing $c_i$ in $\flat^\bullet(P_1)$.
    Hence, $\out(P) \not\transform \gamma_1 \comb \gamma_1 \comb \gamma_1 \comb \gamma_1$.
    It means that $P$ is contentious, and contentious to player~$1$, and $P$ does not belong to $\GOOD^\pref(G^\psi, 1)$; A contradiction.
    Hence, there is no variable $p \in \VAR(\psi)$ and pair of clauses $c_i$ and $c_j$ ($0 \leq i,j \leq m$) such that both $c_i \comb p \in P_1$ and $c_j \comb \lnot p \in P_1$.
    \end{enumerate}
    We conclude that $v^P(\psi) = true$, and $\psi$ is satisfiable.
  \end{proof}

We can now characterize the complexity of $\NE^\pca$ and $\NE^\ppca$.
\begin{theorem}
\label{thm:NE-pca-ppca-coNP-c}
  When $\pref \in \{\pca, \ppca\}$, the problem $\NE^\pref$ is \coNP-complete.
\end{theorem}
  \begin{proof}
    Let $\pref$ be either $\pca$ or $\ppca$.
    $\NE^\pref$ is in \coNP (Prop.~\ref{prop:ne-in-conp-pca-ppca}).
    For hardness, we reduce $\EPD_\emptyset^\pref$ (which is $\NP$-complete according to Lemma~\ref{lemma:epd-empty-npc-pca-ppca}) to $\co\NE^\pref$.
    Let $G = (N,\vec{\gamma},\vec{\epsilon})$ be an RCG, and let $i$ be a player in $N$.
    We construct the RCG $G' = (N, \vec{\gamma}, \vec{\epsilon'})$, where $\epsilon'_i = \epsilon_i$ and $\epsilon'_j = \emptyset$ when $i \not = j$.
    That is, $G'$ is just like $G$, except that every player that is not player~$i$ is dummy. Crucially, all players in $G'$ keep their objectives from $G$. This means that for every profile $P = (C_i, \emptyset, \ldots, \emptyset) \in \ch(G')$, we have $P \in \ISAT(G',i)$ iff $P \in \ISAT(G,i)$, and $P \in \CONT(G',i)$ iff $P \in \CONT(G,i)$, for every player~$i \in N$.
    Clearly, we have $\EPD_\emptyset^\pref(G,i)$ iff $(\emptyset, \ldots, \emptyset) \not \in \NE^\pref(G')$.
  \end{proof}

\subsection{Contention-tolerance}
\label{sec:case-contention-tolerance}

To solve $\NE^\pct$, a non-deterministic algorithm analogous to the one suggested in the proof of Prop.~\ref{prop:ne-in-conp-pca-ppca} can be adopted, yielding a procedure in $\coNP$.
However, we show here that when the players tolerate contentious profiles, then solving the problem $\NE^\pref$ becomes simpler.
The following lemma (a reformulation of \cite[Lemma~5.7]{10.1093/logcom/exaa031}), is the basis of an efficient decision procedure for $\NE^\pct$:
\begin{lemma}
\label{lemma:ne-parsi-affin}
  Let $G = (N, \gamma_1, \ldots, \gamma_n, \epsilon_1, \ldots,
  \epsilon_n)$ be a resource contribution game. We have $P
  \not \in \NE^\pct(G)$ iff $\exists i \in N$ s.t.\ either:
  \begin{enumerate}[leftmargin=*]
  \item $P \not\in \ISAT(G,i)$ and $P_i \not = \emptyset$;
  \item $P \not\in \ISAT(G,i)$ and $(P_{-i},\epsilon_i) \in \ISAT(G,i)$;
  \item $P \in \ISAT(G,i)$ and $\exists A \in P_i$ s.t., $(P_{-i}, P_i \setminus \{A\}) \in \ISAT(G,i)$.
 \end{enumerate}
\end{lemma}
\begin{proof}
    Right to left is immediate. From left to right, suppose $P \not \in \NE^\pct(G)$. So there exists $i \in N$ and $C_i \in \ch_i(G)$ such that $P\prec_i (P_{-i}, C_i)$. There are three cases to consider:
    \begin{enumerate}[label=\roman*, leftmargin=*]
      \item $(P_{-i}, C_i) \not \in \ISAT(G,i)$, $P \not \in \ISAT(G,i)$ and $C_i \subset P_i$;
      \item $(P_{-i}, C_i) \in \ISAT(G,i)$, $P \not \in \ISAT(G,i)$;
      \item $(P_{-i}, C_i) \in \ISAT(G,i)$, $P \in \ISAT(G,i)$ and $C_i \subset P_i$.
      \end{enumerate}
    Suppose~(i) is the case. We have $P \not\in \ISAT(G,i)$, and since $C_i \subset P_i$ we also have $P_i \not = \emptyset$. So item~(1) of the lemma is satisfied.
    Suppose~(ii) is the case. We have, $P \not \in \ISAT(G,i)$. Plus, since $(P_{-i}, C_i) \in \ISAT(G,i)$ we have $\out((P_{-i}, C_i)) \transform \gamma_i$, or $\flat(\gamma_i) \subseteq \flat^\bullet(\out((P_{-i}, C_i)))$. Since $C_i \subseteq \epsilon_i$, we also have $\flat(\gamma_i) \subseteq \flat^\bullet(\out((P_{-i}, \epsilon)))$, and thus  $(P_{-i},\epsilon_i) \in \ISAT(G,i)$. So item~(2) of the lemma is satisfied.
    Suppose~(iii) is the case. We have $P \in \ISAT(G,i)$. Plus, $C_i \subset P_i$, so there is a resource $A$ such that $A \in P_i$ and $A \not \in C_i$. We have $C_i \subseteq P_i \setminus \{A\}$. We also have $(P_{-i}, C_i) \in \ISAT(G,i)$. Applying Fact~\ref{fact:surplus-tolerant-monotony}, we obtain that $(P_{-i}, P_i \setminus \{A\}) \in \ISAT(G,i)$ for some $A \in P_i$. So item~(3) of the lemma is satisfied.
    \end{proof}

This comes as an immediate consequence; a particular case of~\cite[Prop.~4.6]{10.1093/logcom/exaa031}.
\begin{theorem}\label{thm:NE-in-P}
  The problem $\NE^\pct$ is in \PTIME.
\end{theorem}
\begin{proof}
    Lemma~\ref{lemma:ne-parsi-affin} yields a deterministic algorithm. For every $i \in N$, check whether one of the three conditions holds, and if one condition holds, then return false. Otherwise return true. Each condition can be checked efficiently (Prop.~\ref{prop:ptime-satp-good}). Hence, the algorithm runs in polynomial time.
\end{proof}

From Fact~\ref{fact:basic-properties}.\ref{fact:private-cont-at-least-two}, we know that if there is only one player, there cannot be any contention. This follows from Theorem~\ref{thm:NE-in-P}.
\begin{corollary}\label{corr:one-player-P}
When there is only one player the problem $\NE^\pref$ is in \PTIME, even when $\pref \in \{\pca, \ppca\}$.
\end{corollary}

The problem $\NE^\pct$ being tractable, we briefly turn our attention to the problem of deciding\linebreak whether an RCG admits a Nash equilibrium, according to $\pct$.
Given a resource contribution game $G$, the problem $\ENE^\pref$ asks whether $\NE^\pref(G) \not = \emptyset$.

\noindent\begin{minipage}{\linewidth}
\medskip
\hrule height 1pt
\smallskip
\textsf{NASH EQUILIBRIUM EXISTENCE ($\ENE^\pref$)}
\smallskip
\hrule
\begin{description}
\item[{\bf (in)}] A resource contribution game $G$.
\item[{\bf (out)}] $\NE^\pref(G) \not = \emptyset$?
\end{description}
\hrule
\medskip
\end{minipage}

By Theorem~\ref{thm:NE-in-P}, $\ENE^\pct$ is clearly in \NP.
We can show that $\ENE^\pct$ is \NP-hard via a reduction from \SAT.
\begin{theorem}
\label{thm:ENE-ct-NPc}
  The problem $\ENE^\pct$ is \NP-complete.
\end{theorem}
  \begin{proof}
    Membership is immediate from Theorem~\ref{thm:NE-in-P}.
    For hardness, we reduce \SAT. Let $\psi = c_1 \land \ldots \land c_m$ be a 3CNF over the set of propositional variables $\VAR(\psi)$, with $c_i = l_{i,1} \lor l_{i,2} \lor l_{i,3}$, where every literal $l_{i,j}$ is in $\VAR(\psi) \cup \{\lnot p \mid p \in \VAR(\psi)\}$.
    We build the RCG $G^\psi = (\{1\} \cup N_1 \cup N_2, \vec{\gamma}, \vec{\epsilon})$ over $\Res = \{c_i \mid 1 \leq i \leq m  \} \cup \{p, \lnot p, p' \mid p \in \VAR(\psi)\}$ as follows:
    \begin{itemize}[leftmargin=*, wide, labelwidth=!, labelindent=0pt, itemindent=!]
    \item Player~$1$ wants $\gamma_1 = c_1 \comb \ldots \comb c_m$, and is endowed with $\epsilon_1 = \{c_i \comb l_{i,1}, c_i \comb l_{i,2}, c_i \comb l_{i,3} \mid 1 \leq i \leq m\}$.
    \item $N_1 = \{ 2, \ldots, 1 + \Card{\VAR(\psi)} \}$
    \item For every variable $p \in \VAR(\psi)$, the set $N_1$ contains one player $id_1(p)$ with objective $\gamma_{id_1(p)} = p \comb \lnot p \comb p'$, and with the endowment $\epsilon_{id_1(p)} = \{p'\}$.
(The function $id_1$ is an arbitrary bijection between $\VAR(\psi)$ and $N_1$.)
   Every member of $N_1$ wants a resource representing a variable, the negation of the variable, and a distinguished copy of the variable. They are endowed with one occurrence of the distinguished copy.
    \item $N_2 = \{ 2 + \Card{\VAR(\psi)}, \ldots, 1 + 2 \cdot \Card{\VAR(\psi)} \}$     
    \item For every variable $p \in \VAR(\psi)$, the set $N_2$ contains one player $id_2(p)$ with objective $\gamma_{id_2(p)} = p \comb \lnot p \comb p' \comb p'$, and with the endowment $\epsilon_{id_2(p)} = \{p'\}$.
(The function $id_2$ is an arbitrary bijection between $\VAR(\psi)$ and $N_2$.)
      Every member of $N_2$ wants a resource representing a variable, the negation of the variable, and two occurrences of a distinguished variable. They are endowed with one occurrence of the distinguished copy.
    \end{itemize}
    We can show that $\ENE^\pct(G^\psi)$ iff  $\psi$ is satisfiable.
    The idea of the reduction is to have player~$1$ wanting every clause $c_i$. With his endowment, he can achieve it alone, and achieve a profile $P$ such that $P \in \GOOD^\pct(G^\psi, 1)$. But since every $c_i$ in his endowment comes simultaneously with a literal, it might be that both $p$ and $\lnot p$ appear in his choice, for some $p \in \VAR(\psi)$.
    When it happens, it yields a cycle of profitable deviations, by player~$id_1(p)$ in $N_1$ and by player~$id_2(p)$ in $N_2$ alternately. 
 
    When for some $p \in \VAR(\psi)$ both $p$ and $\lnot p$ are in the profile, there are four possible cases, all leading to the existence of a profitable deviation for a player in $N_1$ or $N_2$.
    \begin{enumerate}[leftmargin=*]
    \item There is no occurrence of $p'$. Then, player~$id_1(p)$ currently plays $\emptyset$, and does not achieve his goal. He thus has a profitable deviation in playing $\{p'\}$. We end in the case of item~$2$.
    \item There is only one occurrence of $p'$, provided by player~$id_1(p)$. Then, player~$id_2(p)$ is currently playing his empty choice, and does not achieve his goal. Player~$id_2(p)$ thus has a profitable deviation in playing $\{p'\}$. We end in the case of item~$3$.
    \item There are two occurrences of $p'$. Then, player~$id_1(p)$ currently plays $\{p'\}$, and has a profitable deviation in parsimoniously playing $\emptyset$. We end in the case of item~$4$.
    \item There is only one occurrence of $p'$, provided by player~$id_2(p)$. Player~$id_2(p)$ does not achieve his goal, and has a profitable deviation in parsimoniously playing $\emptyset$. We end in the case of item~$1$.

    \end{enumerate}
    
    When for none of $p \in \VAR(\psi)$, we have both $p$ and $\lnot p$ in the profile, there are three cases.
    \begin{enumerate}[leftmargin=*]
    \item Not every $c_i$ appears in the profile, so player~$1$ has a profitable deviation in playing $\epsilon_1$. We end up in item~$3$ if $\psi$ is trivial, or in one of the cases above if it is not.
    \item Every $c_i$ appears in the profile and there is at least one occurrence of $p'$ for some $p \in \VAR(\psi)$.
      None of player~$id_1(p)$ and player~$id_2(p)$ achieve his goal ($p$ and $\lnot p$ are not both present). Any one of them who provided $p'$ has a profitable deviation in parsimoniously playing $\emptyset$.
    \item Every $c_i$ appears in the profile and there is no occurrence of $p'$ for some $p \in \VAR(\psi)$. Let $v$ be the valuation for $\psi$ such that $v(p) = true$ iff $p$ is in the outcome of the profile. Clearly, $v(\psi) = true$.
    \end{enumerate}

    So a Nash equilibrium exists in $G^\psi$ iff there is a profile $(C_1, \emptyset, \ldots, \emptyset)$, where $\flat^\bullet(C_1)$ contains at least one occurrence of every clause $c_i$, and does not contain both $p$ and $\lnot p$ for any $p \in \VAR(\psi)$. This is possible iff there is a satisfying valuation of $\psi$.
  \end{proof}

\section{RCGs with endowments as bags of atomic resources}
\label{sec:irgbar}

Lemma~\ref{lemma:epd-empty-npc-pca-ppca} is central to prove that $\NE^\pref$ is \coNP-complete in presence of contention-averse preferences. But on inspection, one can observe that the reduction from \SAT builds RCGs where the endowments contain well chosen \emph{bundles} of resources. The proof does not carry over to the case of RCGs where the endowments are bags of atomic resources.
We introduce this class of RCGs.
\begin{definition}
An \emph{RCG with endowments as bags of atomic resources} (RCGBAR) is an RCG\linebreak $G = (N,\gamma_1, \ldots, \gamma_n, \epsilon_1, \ldots, \epsilon_n)$ such that for every $i \in N$, and for every $A \in \epsilon_i$, $A \in \Res$.
\end{definition}
RCGBARs are adequate whenever the resources in the objectives and the resources in the endowments have the same granularity. E.g, on a market, the providers have bottles of wines in their endowment and acquirers have objectives about bottles of wine. The bottle is not conceived as a bundle, but as an atomic resource. On the other hand, as a variant of Example~\ref{ex:bringyourownfood} with you, Bruno, and Carmen, as the objectives are about glasses of wine, it is conceivable that the glasses of wine, and not the bottle, is the relevant resource to consider to model the interaction.
\begin{example}\label{ex:rcgbar-bringyourownfood}
  You participate in a bring-your-own-food cooking party. You bring a bottle of white wine which contains four glasses, and your objective is to enjoy two glasses.  Bruno and Carmen both want a portion of risotto, each portion needing a portion of rice, one onion, and a glass of white wine. Bruno brings three portions of rice, and Carmen brings two onions.
  As before, a glass of wine is represented by $\var{w}$, one portion of rice is represented by $\var{r}$, and an onion is represented by $\var{o}$.
  The game is formally defined as:
$\Gexrcgbarbyof = (\{y, b, c\},\gamma_y = \bundle{\var{w}, \var{w}}, \gamma_b = \bundle{\var{w}, \var{r}, \var{o}}, \gamma_c = \bundle{\var{w}, \var{r}, \var{o}}, \epsilon_y = \{\var{w}, \var{w}, \var{w}, \var{w} \}, \epsilon_b = \{\var{r}, \var{r}, \var{r}\}, \epsilon_c = \{\var{o}, \var{o}\})$.
  For instance, with contention-averse preferences, the profile $(\{\var{w}, \var{w}\}, \emptyset, \emptyset)$ is a Nash equilibrium where you come back home with half a bottle of wine, and $(\{\var{w}, \var{w}, \var{w}, \var{w}\}, \{\var{r}, \var{r}\}, \{\var{o}, \var{o}\})$ is a Nash equilibrium that is good for everyone and you come home empty-handed.
\end{example}
\begin{example}\label{ex:grandfather}
A grandfather owns $4$ parcels of land $\var{x}$ by the river that he could give to his grandchildren. Anna already owns $3$ parcels of $\var{y}$. Brian owns $1$ parcel of $\var{x}$ and $2$ parcels of $\var{z}$. Caren owns $1$ parcel of $\var{z}$. They all agree to pool their assets in the hope to obtain the parcels they like.
Ann and Brian would like $\bundle{\var{x}, \var{y}, \var{z}}$. Caren would like $\bundle{\var{x}, \var{z}}$. The grandfather does not want anything, that is $\emptybundle$.
This yields the game $\Gexrcgbargf$.
In $\NE^\pca(\Gexrcgbargf)$, there are two Nash equilibria  $(\emptyset, \emptyset, \emptyset, \emptyset)$, and $(\{\var{x}, \var{x}\}, \{\var{y}, \var{y}\}, \{\var{x}, \var{z}, \var{z}\}, \{\var{z}\})$.
\end{example}

In this section, we show that in RCGBARs, deciding whether a profile is a Nash equilibrium is in \PTIME, even in presence of (public and private) contention-averse preferences. (In the case of contention-tolerant preferences, it already follows from the general case covered by Theorem~\ref{thm:NE-in-P}.)  We show that there is always a Nash equilibrium when in presence of public contention-averse preferences, while it is not the case for contention-tolerant preferences and private contention-averse preferences.
%


\medskip

The notion of minimal profitable deviation will help simplify the presentation. A minimal profitable deviation is a choice deviation to a strictly preferred profile, that is minimal for multiset inclusion.
\begin{definition}
Let $G = (N,\gamma_1, \ldots, \gamma_n, \epsilon_1, \ldots, \epsilon_n)$ be an RCG, $P$ be a profile in $\ch(G)$, and $i$ a player in $N$. A \emph{minimal profitable deviation} for player~$i$, according to $\pref$, from the profile $P$ is a choice $P'_i \in \ch_i(G)$ such that $P \prec^\pref_i (P_{-i}, P'_i)$ and there is no $P''_i \in \ch_i(G)$ such that $P''_i \subset P'_i$ and $P \prec^\pref_i (P_{-i}, P''_i)$.
\end{definition}

Provided a profitable deviation exists, a minimal profitable deviation also exists. 
Instrumentally, we are going to solve in RCGBARs a direct generalization of $\EPD^\pref_\emptyset$ to arbitrary profiles. For convenience, $\EPD^\pref$ is formulated as a function problem, and returns a witness.

\noindent\begin{minipage}{\linewidth}
\medskip
\hrule height 1pt
\smallskip
\textsf{FIND PROFITABLE DEVIATION ($\EPD^\pref$)}
\smallskip
\hrule
\begin{description}
\item[{\bf (in)}] A resource contribution game $G$, a player~$i$, and a profile $P$.
\item[{\bf (out)}] Is there a profitable deviation according to $\pref$ from $P$ for player~$i$? If yes, find a minimal one.
\end{description}
\hrule
\medskip
\end{minipage}

In a profile of an RCGBAR, a \emph{deficit} from the point of view of a player, is the amount of atomic resources necessary to add to a profile in order to avoid the resource contention.
\begin{definition}
 Let $G = (N,\gamma_1, \ldots, \gamma_n, \epsilon_1, \ldots, \epsilon_n)$ be an RCGBAR, $i$ be a player in $N$, and $P$ be a profile in $\ch(G)$. Let $\pref \in \{\pca, \ppca\}$.
 \begin{itemize}
 \item The set of \emph{contenders} is
   \[\contenders(G, P) = \left\{j \in N \mid P \in \CONT(G,j) \right\} \enspace ;\]
 \item  The multiset of \emph{contended resources} is
   \[
   \contGoals(G, i, P, \pref) = \biguplus_{j \in \contenders(G, P)} 
   \begin{cases}
     \flat\left(\gamma_j\right) & \text{ if } \pref = \pca\\
     \flat\left({\gamma_j}_{\mid \gamma_i}\right) & \text{ if } \pref = \ppca \enspace ;\\
   \end{cases} \]   
 \item The \emph{deficit} from the point of view of player~$i$ in the profile $P$ of the game $G$, is \[\Deficit(G, i, P, \pref) = \contGoals(G, i, P, \pref) \setminus \out(P) \enspace .\]
 \end{itemize}
\end{definition}

\begin{algorithm*}
\caption{
Algorithm to solve $\EPD^\pref(G,i,(P_1, \ldots, P_i = \emptyset, \ldots, P_n))$
in RCGBARs with $\pref \in \{\pca, \ppca\}$
}
\label{alg:RCGBAR-EPD-empty}
\begin{algorithmic}[1]
  \\$\miss = \flat(\gamma_i) \setminus \flat^\bullet(\out((P_{-i}, \emptyset)))$ \Comment{missing resources in $(P_{-i}, \emptyset)$ to satisfy $i$}
  \\\label{alg:test-cannot-potentially-sat-self}if ($\miss \not \subseteq \epsilon_i$): \Comment{$i$ cannot potentially satisfy himself with a unilateral deviation}
  \\\indent\indent\label{line:miss-too-much} return (false,  $\emptyset$) 
  \\\label{alg:potentially-sat}$P_i = \miss$ \Comment{$i$'s current new choice}  
  \\$E_i = \epsilon_i \setminus P_i$ \Comment{current surplus at $i$'s disposition} 
  \\$\deficit = \Deficit(G, i, (P_{-i}, P_i), \pref)$ 
  \\while ($\deficit \not = \emptyset$ and $\deficit \subseteq E_i$): \Comment{there is a deficit that can be covered}
  \\\indent\indent $P_i = P_i \uplus \deficit$
  \\\indent\indent $E_i = E_i \setminus \deficit$  
  \\\indent\indent $\deficit = \Deficit(G, i, (P_{-i}, P_i), \pref)$
  \\\label{alg:last-line}return (~~~$P_i \not = \emptyset$ and $(P_{-i}, P_i) \in \GOOD^\pref(G,i)$~~~,~~~$P_i$~~~) 
\end{algorithmic}
\end{algorithm*}

Trying to make a player potentially satisfied and to cover the deficit in a profile with the resources of the player is exactly what Algorithm~\ref{alg:RCGBAR-EPD-empty} is doing.
Algorithm~\ref{alg:RCGBAR-EPD-empty} \emph{partially} (when the choice from which to find a profitable deviation is empty) solves $\EPD^\pref$ in RCGBARs with contention-averse preferences. 
It returns a pair of values. 
With its first returned value, the algorithm is correct to decide whether a player has a profitable deviation in an RCGBAR, from a profile in which \emph{his} choice is $\emptyset$.
The algorithm's second returned value is a minimal profitable deviation if it exists.
It runs in polynomial time.
\begin{lemma}\label{lem:algo-RCGBAR-EPD-empty}
  Algorithm~\ref{alg:RCGBAR-EPD-empty} returns true on the instance $(G,i,(P_1, \ldots, P_i = \emptyset, \ldots, P_n))$ iff player~$i$ has a profitable deviation from $(P_1, \ldots, P_i = \emptyset, \ldots, P_n)$ in the RCGBAR $G$.
  If it returns true, it also returns a minimal profitable deviation.
  Moreover, the algorithm runs in polynomial time.
  \end{lemma}
\begin{proof}
From left to right, suppose that Algorithm~\ref{alg:RCGBAR-EPD-empty} returns (true, $C_i$) on the instance $(G,i,P)$, where
$P = (P_1, \ldots, P_i = \emptyset, \ldots, P_n)$.
Since the first value returned at line~\ref{alg:last-line} is true,
$C_i \not = \emptyset$ and $(P_{-i}, C_i) \in \GOOD^\pref(G,i)$.
We show that $P \not \in \GOOD^\pref(G,i)$, and therefore that $P \prec^\pref_i (P_{-i}, C_i)$.

Suppose towards contradiction that $P \in \GOOD^\pref(G,i)$. So, $i$ is potentially satisfied and $\miss = \emptyset$. Thus, the test at line~\ref{alg:test-cannot-potentially-sat-self} fails, and line~\ref{alg:potentially-sat} sets $P_i = \emptyset$. The current profile is therefore still $P$.

We have two cases for $\pref \in \{\pca, \ppca\}$.
If $\pref = \pca$, since $P$ is good for $i$, $P$ is not contentious. Therefore, $P \not \in \CONT(G,j)$ for every $j \in N$. It follows that $\contenders(G, P) = \emptyset$, and hence that $\contGoals(G, i, P, \pca) = \emptyset$.
If $\pref = \ppca$, since $P$ is good for $i$, $P \not \in \CONT(G,i)$. Every player~$j \in \contenders(G, P)$ is potentially satisfied in $P$, since $P \in \CONT(G,j)$. Since $P$ is not contentious for $i$, it follows that
$\contGoals(G, i, P, \ppca) = \biguplus_{j \in \contenders(G, P)} \flat\left({\gamma_j}_{\mid \gamma_i}\right) \subseteq \out(P)$.

Thus, in both cases, $\contGoals(G, i, P, \pref) \subseteq \out(P)$. Therefore, $\deficit = \emptyset$ and the loop is not entered. So the execution reaches line~\ref{alg:last-line} with $P_i = \emptyset$, and therefore returns (false, $\emptyset$). A contradiction.

Hence, $P \not \in \GOOD^\pref(G,i)$.
Since $(P_{-i}, C_i) \in \GOOD^\pref(G,i)$, we have
$P \prec^\pref_i (P_{-i}, C_i)$.
Moreover, $C_i \subseteq \epsilon_i$.
Thus, $C_i$ is a profitable deviation for player~$i$ from $P$.

\medskip

From right to left, suppose that player~$i$ has a profitable deviation
from $P = (P_1, \ldots, P_i = \emptyset, \ldots, P_n)$.
Since $P_i = \emptyset$, the existence of a profitable deviation
implies that $P \not \in \GOOD^\pref(G,i)$. Indeed, if $P$ were good
for player~$i$, no profile could be strictly preferred to $P$, since
player~$i$ cannot contribute strictly less than $\emptyset$.

Let $C_i$ be an arbitrary profitable deviation from $P$ for
player~$i$. Since $P \not \in \GOOD^\pref(G,i)$ and $P_i = \emptyset$,
the definition of $\prec^\pref_i$ implies that $(P_{-i}, C_i) \in
\GOOD^\pref(G,i)$.
In particular, player~$i$ is potentially satisfied in $(P_{-i}, C_i)$
and $C_i \subseteq \epsilon_i$.  Since $G$ is an RCGBAR, $\epsilon_i$
and $C_i$ only contain atomic resources.\footnote{In a general RCG,
$\epsilon_i$ may contain non-atomic bundles, whereas $\miss$ and
$\deficit$ are multisets of atomic resources. Covering $\miss$ or
$\deficit$ would require selecting bundles from $\epsilon_i$. Such a
selection need not be unique, and the additional resources in a
selected bundle could create further resource contention. Thus, the
deterministic construction and the inclusions used here do not extend
to general RCGs.}  Therefore, every resource in $\miss$ occurs in
$C_i$, and hence $\miss \subseteq C_i \subseteq \epsilon_i$.
Thus, the test at
line~\ref{alg:test-cannot-potentially-sat-self} fails, and
line~\ref{alg:potentially-sat} sets $P_i = \miss$.
Consequently, player~$i$ is potentially satisfied in $(P_{-i}, P_i)$ and $P_i \subseteq C_i$.

We show that
\begin{equation}
\tag{$\star$}
\label{eq:invariant}
P_i \subseteq C_i
\end{equation}
remains true throughout the execution of the loop. The preceding argument establishes \eqref{eq:invariant} before the first iteration.

Suppose that \eqref{eq:invariant} holds when the condition of the
loop is evaluated, and let $Q = (P_{-i}, P_i)$ and $R = (P_{-i}, C_i)$.
Every player in $\contenders(G, Q)$ is potentially satisfied in $Q$.
Since $P_i \subseteq C_i$, we have $\out(Q) \subseteq \out(R)$, and
therefore every such player is also potentially satisfied in $R$.
Since $R \in \GOOD^\pref(G,i)$, $R$ is not contentious if $\pref =
\pca$, and it is not contentious for player~$i$ if $\pref = \ppca$.
So, $\contGoals(G, i, Q, \pref) \subseteq \out(R)$.
As a consequence, the current value of $\deficit$ satisfies
\begin{align*}
\deficit & =  \Deficit(G, i, Q, \pref)\\
         & =  \contGoals(G, i, Q, \pref) \setminus \out(Q)\\
         &\subseteq  \out(R) \setminus \out(Q)\\
         & =  C_i \setminus P_i \enspace .
\end{align*}
By construction, $E_i = \epsilon_i \setminus P_i$ whenever the condition of the loop is evaluated.
Since $C_i \subseteq \epsilon_i$, we also have $C_i \setminus P_i \subseteq \epsilon_i \setminus P_i = E_i$.
So $\deficit \subseteq E_i$.
Thus, if $\deficit \not = \emptyset$, the loop is entered and the
value $P_i \uplus \deficit$ is assigned to $P_i$, and $P_i \uplus
\deficit \subseteq C_i$. Hence, \eqref{eq:invariant} is preserved by
every iteration of the loop.

At every evaluation of the condition of the loop, \eqref{eq:invariant}
and the preceding argument imply that $\deficit \subseteq E_i$.
Moreover, every iteration adds the non-empty multiset $\deficit$ to
$P_i$ and removes it from $E_i$.  Since $\epsilon_i$ is finite, the
loop terminates.  It cannot terminate with $\deficit \not =
\emptyset$, since in that case $\deficit \subseteq E_i$ and the
condition of the loop would still be true.  Therefore, the loop
terminates with $\deficit = \emptyset$.

Let $Q = (P_{-i}, P_i)$ be the profile obtained when the loop terminates.
Player~$i$ is potentially satisfied in $Q$, since he becomes
potentially satisfied when $P_i$ is initialized with $\miss$, and
potential satisfaction is preserved when resources are added.

Again, we treat the cases of $\pref \in \{\pca, \ppca\}$ separately.
Suppose first that $\pref = \pca$.
If $Q$ were contentious, there would be an atomic resource $A$ such that the players potentially satisfied in $Q$ jointly require more occurrences of $A$ than are available in $\out(Q)$.
Every potentially satisfied player whose goal contains $A$ would would belong to $\contenders(G, Q)$.
Consequently, $A$ would occur more often in
$\contGoals(G, i, Q, \pca)$ than in $\out(Q)$.
This would imply that $\Deficit(G, i, Q, \pca) \not = \emptyset$, a contradiction.
Thus, $Q$ is not contentious.

Suppose now that $\pref = \ppca$.
If $Q \in \CONT(G,i)$, there would be an atomic resource $A$ occurring in $\gamma_i$ such that the players potentially satisfied in $Q$ jointly require more occurrences of $A$ than are available in $\out(Q)$.
Every potentially satisfied player whose goal contains $A$ would belong to $\contenders(G, Q)$.
Since $A$ occurs in $\gamma_i$, restricting those goals to $\gamma_i$ would preserve their occurrences of $A$.
Consequently, $A$ would occur more often in $\contGoals(G, i, Q, \ppca)$ than in $\out(Q)$.
This would imply that $\Deficit(G, i, Q, \ppca) \not = \emptyset$, a contradiction.
Thus, $Q \not \in \CONT(G,i)$.

In both cases, $Q \in \GOOD^\pref(G,i)$. Moreover, $P_i \not = \emptyset$.
Otherwise, $Q = P$, contradicting
$P \not \in \GOOD^\pref(G,i)$.
Therefore, Algorithm~\ref{alg:RCGBAR-EPD-empty} returns (true, $P_i$) at line~\ref{alg:last-line}.

\medskip
We can see that $P_i$ is a minimal profitable deviation.
The profitable deviation $C_i$ was chosen arbitrarily.
By \eqref{eq:invariant}, the choice returned by
Algorithm~\ref{alg:RCGBAR-EPD-empty} satisfies $P_i \subseteq C_i$ for every profitable deviation $C_i$ from $P$.
Suppose that some profitable deviation $C'_i$ satisfied $C'_i \subset P_i$.
Applying \eqref{eq:invariant} with $C'_i$ in place of $C_i$ would give $P_i \subseteq C'_i$, a contradiction.
Hence, $P_i$ is a minimal profitable deviation.

\medskip
Finally, the algorithm runs in polynomial time. Indeed, each set-theoretic operation, and every computation of $\CONT(G)$, $\CONT(G,j)$, $\ISAT(G,j)$, $\GOOD^\pref(G,i)$ is easy (Prop.~\ref{prop:ptime-satp-good}).
The set $E_i$ is initialized with $\epsilon_i \setminus P_i$, and strictly decreases at each iteration.
Since $E_i \subseteq \epsilon_i$, the loop performs at most
$|\epsilon_i|$ iterations.
  \end{proof}

\begin{algorithm*}
\caption{Algorithm to solve $\EPD^\pref(G,i,P)$ for RCGBARs with $\pref \in \{\pca, \ppca\}$} 
\label{alg:RCGBAR-EPD}
\begin{algorithmic}[1]
  \\if ($P \not \in \GOOD^\pref(G,i)$): 
  \\\indent\indent if ($P_i \not = \emptyset$): \Comment{$\emptyset$ is a profitable deviation}
  \\\indent\indent\indent return (true, $\emptyset$)
  \\\indent\indent else:
  \\\indent\indent\indent return $\EPD^\pref(G,i,(P_1, \ldots, P_i = \emptyset, \ldots, P_n))$ \Comment{call Algorithm~\ref{alg:RCGBAR-EPD-empty}}
  \\else:
  \\\indent\indent if ($P_i = \emptyset$): \Comment{profile already good for $i$ and $P_i = \emptyset$}
  \\\indent\indent\indent return (false, $P_i$)
  \\\indent\indent if ($(P_{-i}, \emptyset) \in  \GOOD^\pref(G,i)$): \Comment{$\emptyset$ is a profitable deviation}
  \\\indent\indent\indent return (true, $\emptyset$)
  \\\indent\indent else:  
  \\\indent\indent\indent let $G'$ be the RCG identical to $G$, except for $i$'s endowment being $P_i$
  \\\indent\indent\indent $(b, c) = \EPD^\pref(G',i,(P_1, \ldots, P_{i-1}, \emptyset, P_{i+1}, \ldots, P_n))$ \Comment{call Algorithm~\ref{alg:RCGBAR-EPD-empty}}
  \\\indent\indent\indent return ($b$ and $c \not= P_i$, $c$)
\end{algorithmic}
\end{algorithm*}

%
\begin{theorem}
\label{thm:correctness-algo-ENEempty-RCGBAR}
  In RCGBARs, the problem $\NE^\pref$ is in $\PTIME$, even when $\pref \in \{\pca, \ppca\}$.
\end{theorem}
\begin{proof}
  Let $G = (N,\gamma_1, \ldots, \gamma_n, \epsilon_1, \ldots, \epsilon_n)$.
    To decide whether a profile is a Nash equilibrium, it suffices to check for each player whether they have a profitable deviation.
    This can be done by calling at most $n$ times a decision procedure for the problem $\EPD^\pref$.
    In RCGBAR, it can be solved using Algorithm~\ref{alg:RCGBAR-EPD}.

    Let $P$ be a profile in $G$, and let $i \in N$ be a player.
    The correctness when $P \not \in \GOOD^\pref(G,i)$ is trivial. If $P \in \GOOD^\pref(G,i)$, we distinguish three cases.
    \begin{enumerate}
    \item If $P_i = \emptyset$, there is no profitable deviation.
    \item If $P_i \not = \emptyset$ but $(P_{-i}, \emptyset) \in \GOOD^\pref(G,i)$, then $\emptyset$ is a profitable deviation.
    \item If $P_i \not = \emptyset$ and $(P_{-i}, \emptyset) \not\in \GOOD^\pref(G,i)$, and since $P \in \GOOD^\pref(G,i)$, if a profitable deviation exists then it is not $\emptyset$, and it is (strictly) included in $P_i$. In other words, we must find a profitable deviation $P'_i$ such that $\emptyset \subset P'_i \subset P_i$. We thus define a new game $G' = (N',\gamma'_1, \ldots, \gamma'_n, \epsilon'_1, \ldots, \linebreak \epsilon'_n)$, where $N' = N$, $\gamma'_j = \gamma_j$, for all $j \in N$, and $\epsilon'_j = \epsilon_j$ when $j \not = i$, and $\epsilon'_i = P_i$. The algorithm then calls Algorithm~\ref{alg:RCGBAR-EPD-empty} to find a profitable deviation for player~$i$ from the profile $(P_{-i}, \emptyset)$. The algorithm then returns true when Algorithm~\ref{alg:RCGBAR-EPD-empty} returns $(b,c)$, $b$ is true, and the profitable deviation $c$ found (profitable from $(P_{-i}, \emptyset)$ but not necessarily from $(P_{-i}, P_i)$) is not the original choice $P_i$.
    \end{enumerate}
    We thus need to make only a polynomial number of calls to the subprocedure of Algorithm~\ref{alg:RCGBAR-EPD-empty}. Moreover, by Prop~\ref{prop:ptime-satp-good}, we know that deciding which instances to call is easy.
    From Lemma~\ref{lem:algo-RCGBAR-EPD-empty}, we have that
    Algorithm~\ref{alg:RCGBAR-EPD-empty} is correct to decide whether a player in an RCGBAR has an incentive to deviate from a profile where their choice is $\emptyset$, and it runs in polynomial time.
    Hence, in RCGBAR, the problem $\EPD^\pref$ can be solved in polynomial time, and thus $\NE^\pref$ is also in $\PTIME$.
  \end{proof}

\begin{algorithm}
  \caption{Algorithm to find a Nash equilibrium in an RCGBAR, with $\pref = \pca$}
\label{alg:find-Nash-RCGBAR-pca}
\begin{algorithmic}[1]
\\$P = (\emptyset, \ldots, \emptyset)$
\\while ($P \not \in \NE^\pca(G)$):
\\\label{line:players-iter}\indent\indent for each $i \in N$, from $1$ to $n$, do:
\\\indent\indent\indent $(b, c) = \EPD^\pref(G,i, (P_{-i}, P_i))$ \Comment{call Algorithm~\ref{alg:RCGBAR-EPD}}
\\\indent\indent\indent if $b$:
\\\indent\indent\indent\indent $P_i = c$
\\return $P$
\end{algorithmic}
\end{algorithm}

The following can be proved by adopting a special best-response dynamics~\cite{MILCHTAICH1996111} as described in Algorithm~\ref{alg:find-Nash-RCGBAR-pca}.
\begin{theorem}
\label{thm:contention-averse-end-bags-always-NE}
  Let $G$ be an RCGBAR. The set $\NE^\pca(G)$ is non-empty, and finding a profile in it can be done in polynomial time.
  \end{theorem}

  \begin{proof}
    Algorithm~\ref{alg:find-Nash-RCGBAR-pca} implements a best-response dynamics from the empty profile, with the players taking turns. 
    One starts in the profile $P_0 = (\emptyset, \ldots, \emptyset)$.
    During the dynamics $P_0, P_1, \ldots$, at each step $k$:
    \begin{enumerate}
    \item $P_k$ is not contentious,
    \item $\out(P_k) = \biguplus_{i \in \SATP(G,P_k)} \gamma_i$,
    \item $\forall i \not\in \SATP(G,P_k), {P_k}_i = \emptyset$.
    \end{enumerate}
    To see this:
    \begin{itemize}
    \item $(\emptyset, \ldots, \emptyset)$ is not contentious (Fact~\ref{fact:basic-properties}.\ref{fact:trivial-uncontentious}). If $i \in \SATP(G,(\emptyset, \ldots, \emptyset))$, then $\gamma_i = \emptybundle$. Also, $\out(P_0) = \emptyset$. So $\out(P_0) = \biguplus_{i \in \SATP(G,P_0)} \gamma_i$. Clearly $\forall i \not\in \SATP(G,P_k), {P_k}_i = \emptyset$.
    \item Suppose $P_k$ satisfies the conditions 1--3. 
      Between $P_k$ and $P_{(k+1)}$, some player~$j$ has changed their choice from $\emptyset$ to ${P_{(k+1)}}_j \not = \emptyset$. It must be that $P_{(k+1)} \in \GOOD(G,j)$. So $P_{(k+1)}$ is not contentious.\footnote{This would not be the case with private contention-averse preferences! One profitable deviation might not be privately contentious to the deviator but still be privately contentious to some other players.} This and Fact~\ref{fact:surplus-tolerant-monotony} imply $\biguplus_{i \in \SATP(G,P_{(k+1)})} \gamma_i \subseteq \out(P_{(k+1)})$. Moreover, the deviations returned by Algorithm~\ref{alg:RCGBAR-EPD} are minimal (this follows from Lemma~\ref{lem:algo-RCGBAR-EPD-empty}), and because $\out(P_k) \subseteq \biguplus_{i \in \SATP(G,P_k)} \gamma_i$, we also have $\out(P_{(k+1)}) \subseteq \biguplus_{i \in \SATP(G,P_{(k+1)})} \gamma_i$. Finally, since\linebreak 
      $\forall i \not\in \SATP(G,P_k), {P_k}_i = \emptyset$, player~$j$'s choice is the only change from $P_k$ to $P_{(k+1)}$ and $j \in \SATP(G, P_{(k+1)})$, we also have $\forall i \not\in \SATP(G,P_{(k+1)}), {P_{(k+1)}}_i = \emptyset$.
    \end{itemize}


    At every step~$k$, (1)~$P_k$ is not contentious, (2)~$\out(P_k) = \biguplus_{i \in \SATP(G,P_k)} \gamma_i$, and (3)~all players that are not potentially satisfied in $P_k$ are not contributing anything. It implies that the players never have an incentive to remove resources during the iterative process.
    If a profile is good for a player at some point in the iteration, every subsequent profile will also be good for this player.

  The amount of resources in the current profile strictly increases at every `while' iteration.
  This process will eventually terminate in a Nash equilibrium.
  Because of Lemma~\ref{lem:algo-RCGBAR-EPD-empty}, it runs in polynomial time.
  \end{proof}

Observe that line~\ref{line:players-iter} of Algorithm~\ref{alg:find-Nash-RCGBAR-pca} is only an arbitrary way to force a deterministic procedure. But this is not necessary, as any \emph{minimal} deviation, by any player, could be taken at any time.
Using variations of Algorithm~\ref{alg:find-Nash-RCGBAR-pca}, we can possibly find several Nash equilibria, simply by modifying the order of the players. Each permutation of the players would result in a Nash equilibrium, which must not be unique.
On the other hand, not all Nash equilibria can be found only by permuting the players and running Algorithm~\ref{alg:find-Nash-RCGBAR-pca}.
\begin{example}\label{ex:NE-RCGBAR-permutations}
Let the RCG $G$ be  $(\{1,2,3\},  \gamma_1,\gamma_2,\gamma_3, \epsilon_1, \epsilon_2, \epsilon_3)$.
Player~$1$, is endowed with $\epsilon_1 = \{A,B\}$ and has the objective $\gamma_1 = A$. Player~$2$ and player~$3$, both identical, are endowed with $\epsilon_2 = \epsilon_3 = \{A, B, B\}$ and have the objective $\gamma_2 = \gamma_3 = B$.

Let us assume $\pref = \pca$.
We start from the profile $(\emptyset, \emptyset, \emptyset)$ consider the natural order $(1,2,3)$ of $N$.
Following Algorithm~\ref{alg:find-Nash-RCGBAR-pca},
player~$1$ deviates from $(\emptyset, \emptyset, \emptyset)$ to $(\{A\}, \emptyset, \emptyset)$. Then, player~$2$ deviates to $(\{A\}, \{B,B\}, \emptyset)$. Player~$3$ does not deviate from $(\{A\}, \{B,B\}, \emptyset)$, which is a Nash equilibrium.

If instead we permute player~$2$ and player~$3$, and thus consider the order $(1,3,2)$ of $N$, Algorithm~\ref{alg:find-Nash-RCGBAR-pca} results in the Nash equilibrium $(\{A\}, \emptyset, \{B,B\})$.

Now, the profile $(\{A\}, \{B\}, \{B\})$ is a Nash equilibrium in $G$. However, there is no starting permutation of the set of players for which Algorithm~\ref{alg:find-Nash-RCGBAR-pca} would result in it.
\end{example}

The Nash equilibrium resulting from Algorithm~\ref{alg:find-Nash-RCGBAR-pca} might not be efficient, if for instance, efficiency is measured by the number of satisfied players.\footnote{Algorithm~\ref{alg:find-Nash-RCGBAR-pca} searches for a Nash equilibrium starting from $(\emptyset, \ldots, \emptyset)$. We could think of a different search starting from $(\epsilon_1, \ldots, \epsilon_n)$. Non potentially satisfied players must minimally deviate to $\emptyset$, and potentially satisfied players would remove some resources, trying to remain potentially satisfied and ``reducing'' the resource contention.
}
\begin{example}
Let the RCG $G$ be  $(\{1,2\},  \gamma_1, \gamma_2, \epsilon_1, \epsilon_2)$.
Player~$1$ and player~$2$ are identical: they are endowed with $\epsilon_1 = \epsilon_2 = \{A\}$ and have the objective $\gamma_1 = \gamma_2 = A$.
We have $\NE^\pca(G) = \{(\emptyset, \emptyset), (\{A\}, \{A\})\}$. Algorithm~\ref{alg:find-Nash-RCGBAR-pca} will return $(\emptyset, \emptyset)$. (Inverting the two identical players player~$1$ and player~$2$ as suggested in Example~\ref{ex:NE-RCGBAR-permutations} has no effect.)
\end{example}

\paragraph{The case of contention-tolerant preferences.}
In presence of contention-tolerant preferences, a pure Nash equilibrium is not guaranteed to exist, even when the endowments are bags of atomic resources.
Algorithm~\ref{alg:find-Nash-RCGBAR-pca} to find a Nash equilibrium,
does not work in general
if (mistakenly) applied with con\-tention-tolerant preferences.
Indeed, when player~$i$ profitably minimally deviates, it might be the case that some player now has an incentive to withhold some resources from his current choice.

\begin{proposition}
\label{prop:no-nash-irgbar-pct}
Some RCGBARs do not admit any pure Nash equilibrium in presence of contention-tolerant preferences.
\end{proposition}
\begin{proof}
Consider the RCGBAR $G$, where $\gamma_1 = A$, $\gamma_2 = A \comb A$, $\epsilon_1 = \{A\}$, $\epsilon_2 = \{A\}$.
From $(\emptyset, \emptyset)$, player~$1$ has a profitable deviation to $(\{A\}, \emptyset)$, from which player~$2$ has a profitable deviation to $(\{A\}, \{A\})$. By parsimony, and because the players are contention-tolerant, player~$1$ has a profitable deviation from $(\{A\}, \{A\})$ to $(\emptyset, \{A\})$.
The profile $(\emptyset, \{A\})$ is not a Nash equilibrium either because player~$2$ has a profitable deviation to $(\emptyset, \emptyset)$.
\end{proof}

\paragraph{The case of private contention-averse preferences.}
Algorithm~\ref{alg:find-Nash-RCGBAR-pca} may not terminate in the case of $\ENE^\ppca$, too.
\begin{example}
An example in which Algorithm~\ref{alg:find-Nash-RCGBAR-pca} does not terminate when the preferences are private contention-aversion is the $3$-player RCG $G$, where $\epsilon_1 = \{B, B\}$, $\epsilon_2 = \emptyset$, $\epsilon_3 = \{A, A\}$, $\gamma_1 = B$, $\gamma_2 = A \comb B$, and $\gamma_3 = A$.
Let us assume $\pref = \ppca$.
%
Algorithm~\ref{alg:find-Nash-RCGBAR-pca} first sees player~$1$ deviating from $(\emptyset, \emptyset, \emptyset)$, to $(\{B\}, \emptyset, \emptyset)$.
Then, it will cycle through this series of profiles, where each player takes turns to deviate to a minimal profitable deviation when it exists: turn of player~$2$,
$(\{B\}, \emptyset, \emptyset)$, turn of player~$3$,
$(\{B\}, \emptyset, \{A, A\})$, turn of player~$1$,
$(\emptyset, \emptyset, \{A, A\})$, turn of player~$2$,
$(\emptyset, \emptyset, \{A, A\})$, turn of player~$3$,
$(\emptyset, \emptyset, \{A\})$, turn of player~$1$,
$(\{B, B\}, \emptyset, \{A\})$, turn of player~$2$,
$(\{B, B\}, \emptyset, \{A\})$, turn of player~$3$,
$(\{B, B\}, \emptyset, \emptyset)$, turn of player~$1$, and back to
$(\{B\}, \emptyset, \emptyset)$.

Moving from $(\{B\}, \emptyset, \emptyset)$ to
$(\{B\}, \emptyset, \{A, A\})$, player~$3$ ensures that he is potentially satisfied and that the resulting profile is not privately contentious from his point of view. But the profile is contentious from the two other players' point of view. This is the reason why the algorithm is not fit for $\ppca$ preferences. Then, instead of monotonically increasing the amount of resources in the profile, we see next player~$1$ removing his contribution as his minimal profitable deviation. (Before the deviation, the profile is not good for him, hence choosing $\emptyset$ is profitable. In this case, outside the dictate of the algorithm, he could have covered the deficit, deviating to $\{B, B\}$.)

\end{example}
In fact, in presence of private contention-averse preferences, a pure Nash equilibrium is not guaranteed to exist, even when the endowments are bags of atomic resources.
\begin{proposition}
\label{prop:no-nash-irgbar-ppca}
Some RCGBARs do not admit any pure Nash equilibrium in presence of private\linebreak contention-averse preferences.
\end{proposition}
\begin{proof}
Let us assume $\pref = \ppca$. Let $G$ be the $4$-player RCGBAR, where
$\epsilon_1 = \emptyset$, $\epsilon_2 = \emptyset$, $\epsilon_3 = \{B, B\}$, $\epsilon_4 = \{A, A\}$,
$\gamma_1 = \bundle{A,A,B,B}$,
$\gamma_2 = \bundle{A, B}$,
$\gamma_3 = B$,
and $\gamma_4 = A$.
The game $G$ does not admit any Nash equilibrium.
%
The following statements hold:
\begin{itemize}
\item $(\emptyset,\emptyset,\emptyset,\emptyset) \prec^\ppca_3 (\emptyset,\emptyset,\{B\},\emptyset)$,
\item $(\emptyset,\emptyset,\emptyset,\{A\}) \prec^\ppca_3 (\emptyset,\emptyset,\{B,B\},\{A\})$,
\item $(\emptyset,\emptyset,\emptyset,\{A,A\}) \prec^\ppca_4 (\emptyset,\emptyset,\emptyset,\{A\})$,
\item $(\emptyset,\emptyset,\{B\},\emptyset) \prec^\ppca_4 (\emptyset,\emptyset,\{B\},\{A,A\})$,
\item $(\emptyset,\emptyset,\{B\},\{A\}) \prec^\ppca_3 (\emptyset,\emptyset,\emptyset,\{A\})$,
\item $(\emptyset,\emptyset,\{B\},\{A,A\}) \prec^\ppca_3 (\emptyset,\emptyset,\emptyset,\{A,A\})$,
\item $(\emptyset,\emptyset,\{B,B\},\emptyset) \prec^\ppca_3 (\emptyset,\emptyset,\{B\},\emptyset)$,
\item $(\emptyset,\emptyset,\{B,B\},\{A\}) \prec^\ppca_4 (\emptyset,\emptyset,\{B,B\},\emptyset)$,
\item $(\emptyset,\emptyset,\{B,B\},\{A,A\}) \prec^\ppca_3 (\emptyset,\emptyset,\emptyset,\{A,A\})$.
\end{itemize}

\end{proof}

\section{Other variants with a tractable problem \NE}
\label{sec:case-other-variants}

Before concluding this paper,
we present variants for which deciding whether a profile is a Nash equilibrium can be done in polynomial time, even in presence of (public or private) contention-aversity.
Compared to RCGBAR, their tractability remains very straightforward observations.
Still, these variants can be usefully exploited in specific practical applications. They also serve the theoretical interest of understanding the sources of computational complexity in the problem $\NE^\pref$.

\subsection{RCGs with independent objectives}

We now consider the restricted class of RCGs where the goals of the players are independent.
\begin{definition}
An \emph{RCG with independent objectives} (RCGIO) is an RCG $G = (N,\gamma_1, \ldots, \gamma_n, \epsilon_1, \ldots, \epsilon_n)$ such that for every $i,j \in N$, if $i \not = j$ then $\flat(\gamma_i) \cap \flat(\gamma_j) = \emptyset$.
\end{definition}

When the objectives of the players are independent, there is no place whatsoever for contention.
\begin{lemma}\label{lem:indep-goals-no-contention}
  If $G$ is an RCG with independent objectives, then there is no contentious profile in $G$.
\end{lemma}
The converse of Lemma~\ref{lem:indep-goals-no-contention} does not hold. Take $N = \{1,2\}$, $\gamma_1 = A$, $\gamma_2 = A$, $\epsilon_1 = \emptyset$, $\epsilon_2 = \emptyset$. The unique profile $(\emptyset, \emptyset)$ is not contentious (Fact~\ref{fact:basic-properties}). It is also the case with a non-trivial set of profiles, e.g., with $\epsilon'_1 = \{A\comb A\}$, $\epsilon'_2 = \{A\comb A, B\}$, where none of the profiles are contentious.

\begin{proposition}
\label{NE-RCG-ind-Ptime}
  In RCGs with independent objectives, for every kind of preferences \pref, the problem $\NE^\pref$ is in \PTIME.
  \end{proposition}
  \begin{proof}
    By Lemma~\ref{lem:indep-goals-no-contention}, we know that resource contention is irrelevant in an RCG $G$ with independent objectives. Indeed, $\NE^\pct(G) = \NE^\pca(G) = \NE^\ppca(G)$. We can thus decide all three problems $\NE^\pref$ as if $\pref = \pct$. We conclude using Theorem~\ref{thm:NE-in-P}, whose proof provides an algorithm for all the present cases.
  \end{proof}

Independent objectives are not sufficient to guarantee the existence of Nash
equilibria.
\begin{proposition}\label{prop:no-nash-ind}
  For every kind of preferences $\pref \in \{\pct, \pca, \ppca\}$, there is an RCG with independent objectives $G$ such that $\NE^\pref(G) = \emptyset$.
\end{proposition}   
  \begin{proof}
    Consider the three-player RCG $G$ , where $\epsilon_1 = \{A \comb B\}$,
    $\epsilon_2 = \{B \comb C\}$,
    $\epsilon_3 = \{A \comb C\}$,
    and $\gamma_1 = A$, $\gamma_2 = B$, and $\gamma_3 = C$.
    The game $G$ is is an RCG with independent objectives.
    We have $\NE^\pct(G) =\NE^\pca(G) =\NE^\ppca(G) = \emptyset$.       
    \end{proof}
However, restricting the endowments to bags of atomic resources, $\ENE^\pref$ becomes a trivial problem. The following proposition characterizes the existence of one and only one Nash equilibrium.
\begin{proposition}\label{prop:NE-characterize-RCGBAR-IO}
Let $G$ be an RCGBAR with independent objectives. For every kind of preferences $\pref \in \{\pct, \pca, \ppca\}$, 
$P \in \NE^\pref(G)$ iff for every $i \in N$, we have
\[
P_i =
\begin{cases}
  \gamma_i & \text{when $\gamma_i \subseteq \epsilon_i$}\\
  \emptyset & \text{otherwise} \enspace .
\end{cases}
\]
\end{proposition}

\subsection{RCGs with explicit choices}
\label{sec:expl-choices}

A source of computational complexity to decide profitable deviations and Nash equilibria comes from the fact that the number of distinct choices a player can make in an RCG is exponential with the size of his endowment. Each player can play unrestrained any subset of his endowment. This is not always a good assumption either. A player might be subject to some limitations about the combinations of resources it provides, either in numbers (e.g., no more than three candies), or in quality (e.g., never molecular hydrogen and oxygen together).

\begin{definition}
An \emph{individual resource game with explicit choices} (RCGEC) is a tuple\linebreak $G = (N,\gamma_1, \ldots, \gamma_n, \ch_1, \ldots, \ch_n)$ where:
\begin{itemize}[leftmargin=*]
\item $N = \{1, \ldots, n\}$ is a finite set of players;
\item $\gamma_i$ is a resource bundle ($i$'s goal, or objective);
\item $\ch_i$ is a finite set of resource bags ($i$'s choices).
\end{itemize}
\end{definition}
In an RCG $G$ we defined the set $\ch_i(G)$ of possible choices of player~$i$ as the set of multisets (Section~\ref{sec:irg-dp}). In an RCGEC, the sets of choices $\ch_i$ are explicitly specified in the model.
This is useful and not unseen before, indeed. Congestion Games~\cite{Rosenthal1973}, archetypal games with resources, are often presented with such explicit actions.

%

It is easy to see that RCGECs are more general than RCGs. To every RCG $G = (N,\gamma_1, \ldots, \gamma_n, \epsilon_1, \ldots, \epsilon_n)$, corresponds the RCGEC $G' = (N,\gamma_1, \ldots, \gamma_n, \mathcal{P}(\epsilon_1), \ldots, \mathcal{P}(\epsilon_n))$, where $\mathcal{P}(\epsilon_i)$ is the powerset of the multiset $\epsilon_i$.

\begin{proposition}\label{prop:NE-expl-goals-P}
  In RCGs with explicit choices, for every kind of preferences \pref, the problem $\NE^\pref$ is in \PTIME.
\end{proposition}

\subsection{RCGs with bounded-sized endowments}

We observed in  Section~\ref{sec:expl-choices}, that a source of complexity in
finding profitable deviations comes from the number of possible
choices of a player being exponential in the size of his endowment.
We can take steps upon this observation, albeit in a different direction.

\begin{definition}
A \emph{resource contribution game with $k$-bounded endowments} ($k$-RCGBSE) is a resource contribution game $G = (N,\gamma_1, \ldots, \gamma_n, \epsilon_1, \ldots, \epsilon_n)$, where $k$ is an integer such that $\Card{\epsilon_i} \leq k$.
\end{definition}
Using $k$-bounded endowments, the number of resource bundles in the endowments is bounded. But the number of atomic resources that make up these bundles is not bounded. Indeed, in the definition above, the size of $\flat^\bullet(\epsilon_i)$ can still be arbitrarily large.

In $k$-RCGBSEs, the number of choices available to a player is bounded, although exponentially large wrt.\ the size of his endowment.
\begin{proposition}\label{prop:NE-bounded-endow-P}
  Let $k$ be a fixed integer. In the class of RCGs with $k$-bounded endowments, for every kind of preferences \pref, the problem $\NE^\pref$ is in \PTIME.
\end{proposition}

\section{Conclusions}
\label{sec:conclusions}

We introduced a class of non-cooperative games, where players contribute resources, and consume resources. They are largely inspired from the individual resource games of~\cite{Tr16ijcai}. Our main conceptual contribution has been to introduce the notion of resource contention, and to define resource-averse preferences. Our main technical contributions have been about the computational aspects of Nash equilibria in this setting.
We then studied the complexity of the problem of deciding whether a profile is a Nash equilibrium ($\NE^\pref$), and of some interesting cases of the problem of finding a Nash equilibrium ($\ENE^\pref$).

The main results about the problem $\NE^\pref$ are summarized in Table~\ref{tab:ne-summary}. (The case of resource-tolerant preferences is only an adaption form~\cite{10.1093/logcom/exaa031}.)
Also notably, in Theorem~\ref{thm:contention-averse-end-bags-always-NE}, using a special best-response dynamics, we establish that in RCGBARs, with public contention-averse preferences, there is always a pure Nash equilibrium and finding one can be done in polynomial time. This puts this kind of preferences apart; even in RCGBAR, a pure Nash equilibrium may not exist in presence of contention-tolerant preferences and \emph{private} contention-averse preferences.

The problem $\NE^\pref$ is also tractable for every kind of preferences in RCG with independent objectives (Prop.~\ref{NE-RCG-ind-Ptime}), with explicit choices (Prop.~\ref{prop:NE-expl-goals-P}), and with bounded-size endowments (Prop.~\ref{prop:NE-bounded-endow-P}). A Nash equilibrium need not exist in any of these classes.


\begin{table}
\begin{center}
\begin{tabular}{lll}
\toprule
$\pref$  &  RCG & RCGBAR\\
 \midrule
$\pct$ & \PTIME (Th.~\ref{thm:NE-in-P}, \cite{10.1093/logcom/exaa031}) & \PTIME (Th.~\ref{thm:NE-in-P}, \cite{10.1093/logcom/exaa031})\\
$\pca$ & \coNP-c (Th.~\ref{thm:NE-pca-ppca-coNP-c}) & \PTIME (Th.~\ref{thm:correctness-algo-ENEempty-RCGBAR})\\
$\ppca$ & \coNP-c (Th.~\ref{thm:NE-pca-ppca-coNP-c}) & \PTIME (Th.~\ref{thm:correctness-algo-ENEempty-RCGBAR})\\
\bottomrule
\end{tabular}
\end{center}
\caption{\label{tab:ne-summary} Complexity of the problem $\NE^\pref$ in RCG and RCGBAR when $\pref$ is contention-tolerant ($\pct$), public contention-averse ($\pca$), and private contention-averse ($\ppca$).}
\end{table}

\paragraph{Related work.}
\label{sec:related-work}




Deciding whether a profile is a Nash equilibrium and finding a Nash equilibrium
are natural problems that have been largely studied in a variety of settings.

Nash's seminal articles~\cite{Nash50,Nash51} demonstrate that a \emph{mixed-strategy} equilibrium that now bears his name always exists.
A celebrated result of computational game theory establishes that finding a mixed-strategy Nash equilibrium
 in traditionally-defined non-cooperative games (with an explicit set of pure strategies and utilities)
is $\sf{PPAD}$-complete~\cite{DBLP:journals/siamcomp/DaskalakisGP09}; On the other hand, checking whether a profile is a Nash equilibrium is an easy task.

\emph{Pure-strategy} Nash equilibrium synthesis in games with temporal objectives has also attracted a lot of attention in theoretical computer science and multiagent systems, e.g.,~\cite{DBLP:conf/csl/ChatterjeeMJ04,DBLP:conf/tacas/FismanKL10,DBLP:journals/corr/BouyerBMU15,CoDiOuTr21aamas,DBLP:journals/acta/GutierrezMPRSW21,10.1145/3467001.3467003,CoDiOuTr23aamas}, where a significant part of the works focuses on combined qualitative and quantitative objectives. The complexity varies widely depending on the exact models and the specification language of the objectives. Model checking techniques have been successfully adapted to the verification of Nash equilibria~\cite{DBLP:conf/aaai/WooldridgeGHMPT16,GUTIERREZ201553}.

\medskip

These problems have been studied in models of games which share closer similarities with ours.

In congestion games (CGs)~\cite{Rosenthal1973,ieong05}, the players choose a set of resources to use, and their utility depends on the `delays' of shared resources, which depend on the number of players choosing them. Although delays can be arbitrarily large, the resources do not become oversubscribed, and they are not subject to contention.
%
Despite some apparent similarities between RCGs and CGs, they are rather superficial.
Players in CGs are not producers of resources, and do not have endowments \emph{per se}.
Players' actions in CGs consist in choosing a subset of an already available common pool of resources to use.
In RCGs, players are consumers but also producers of resources; their actions consist in making resources available in the common pool.
%
Unlike CGs, the existence of a pure NE is not guaranteed in general in RCGs.
%
Finding a Nash equilibrium in CGs is $\sf{PLS}$-complete~\cite{10.1145/1007352.1007445}, and a socially optimal Nash equilibrium can be found in polynomial time in singleton CGs~\cite{ieong05}.

Somehow, also in boolean games~\cite{Harrenstein:2001:BG:1028128.1028159,DBLP:conf/ecai/BonzonLLZ06,harrenstein2015aamas} the players produce and consume `resources'.
Each player controls a set of boolean variables and produces truth values which can be used without restriction towards the boolean goal of all the players.
But we do not think that there are immediate natural correspondences between RCGs and boolean games. Boolean states of affairs are non-rivalrous resources by nature, so a player using the truth of a propositional variable to achieve their boolean goal does not prevent another player to also do so.
As such, resource contention is absent from boolean games.
So players in boolean games could as well be contention-tolerant. As in boolean games, we could force the endowments to be non-overlapping (for exclusive control over a resource). Under these conditions, a connection would then exist if we allowed the players in our games to have preferences about the absence of a resource.
The main focus of boolean games is also on pure-strategy Nash equilibria. Deciding whether a profile is a Nash equilibrium in a boolean game is \coNP-complete and deciding whether one exists is $\sf{\Sigma_2^p}$-complete.


There is of course a neat relationship with~\cite{Tr16ijcai}, which defines parameterized individual resource games. 
The parameter can be any resource-sensitive logic, specifically variants of Linear Logic~\cite{girard1987}. The language of the logic is then used to represent resources of varying complexity.
Our games are like those of~\cite{Tr16ijcai}, called individual resource games (IRGs), parameterized with a logic corresponding to the fragment $A ::= \mathbf{1} \mid p \mid A \otimes A$ where $p$ is an atomic variable, $\mathbf{1}$ is the empty resource bundle $\emptybundle$, and a formula $A \otimes B$ corresponds to the resource bundle $A \comb B$. Our resource transformation $Bag \transform Bundle$ is nothing but the entailment $Bag \vdash Bundle$ in the Linear Logic.
%
The games from~\cite{Tr16ijcai} are studied only with contention-tolerant players.
In IRGs, the complexity of deciding whether a profile is a Nash equilibrium depends on the variant of Linear Logic that is used to represent the resources. It varies from \PTIME (when using MULT like in our case with contention-tolerant preferences), to undecidable (when using full Linear Logic).


\medskip

Outside of non-cooperative game theory, our games have some connections with other models of resource-sensitive interactions of agents

The classes of games presented in~\cite{Dunne201020,tr18aaai,Su2020ACM-taas}
are models where the players are both consumers and producers of resources. In these games, the players have resource endowments which can be combined so as to achieve resource objectives. These games are cooperative, and the authors study the computational aspects of cooperative solution concepts. %

The notion of competition over resources is ubiquitous in social choice (e.g., \cite{CramptonShohamSteinberg06,DBLP:journals/informaticaSI/ChevaleyreDELLMPPRS06,Brandt:2016:HCS:3033138,bouveretChap12,Airiau2014,aziz19}). The models of our games bear a resemblance with combinatorial exchanges~\cite{DBLP:conf/atal/KothariSS04} and with mixed multi-unit combinatorial auctions (MMUCAs)~\cite{Cerquides:2007:BLW:1625275.1625473,Giovannucci2010}, where the agents can be both sellers and buyers. In MMUCAs, as in RCGs, bundles of goods can be transformed into different bundles of goods. The study of MMUCAs focuses on determining the sequences of bids to be accepted by an auctioneer.

\paragraph{Open problems and future work.}
We have focused more closely on the problem of deciding whether a profile is a pure Nash equilibrium. In the future work, we will look more closely at the problem of finding Nash equilibria. In particular, for contention-averse preferences, $\ENE^\pref$ is in~$\sf{\Sigma_2^p}$, but no lower bound is known. In RCGBARs, we proved that $\ENE^\pca$ is trivial and that finding one can be done in polynomial time (Theorem~\ref{thm:contention-averse-end-bags-always-NE}). But characterizing the complexity of $\ENE^\ppca$ and $\ENE^\pct$ remains to be done; although they are both in $\NP$ (Theorem~\ref{thm:ENE-ct-NPc}, and as a corollary of Theorem~\ref{thm:correctness-algo-ENEempty-RCGBAR}).
It would also be interesting to consolidate existing results. The hardness of $\NE^\pca$ and $\NE^\ppca$ relies on Lemma~\ref{lemma:epd-empty-npc-pca-ppca}. But it does not carry through when we have a bounded supply of players, or when the size of goals is bounded.

\medskip
Working with our definition of preferences (Definition~\ref{def:preferences}) has the advantage to make all assumptions about the motivations of the players explicit. Undeniably, a utility model would be useful for further investigations of RCGs.
However, $\prec^\pref_i$ does not satisfy negative transitivity, and thus cannot be represented by a utility function~\cite{Kreps88}.
We can redefine $Q \prec^\pref_i P$ as:
\begin{enumerate}[leftmargin=*]
\item $P \not \in \GOOD^\pref(G,i)$, $Q \not \in \GOOD^\pref(G,i)$ and $|P_i| < |Q_i|$;
\item $P \in \GOOD^\pref(G,i)$, and $Q \not \in \GOOD^\pref(G,i)$;
\item $P \in \GOOD^\pref(G,i)$, $Q \in \GOOD^\pref(G,i)$ and $|P_i| < |Q_i|$.
\end{enumerate}
We can also define for every player~$i$, and $P \in \ch(G)$:
\[u^\pref_i(P) = \begin{cases}
        2\cdot |\epsilon_i| + 1 - |P_i| & \text{when } P \in \GOOD^\pref(G, i)\\
        |\epsilon_i| - |P_i| & \text{otherwise} \enspace .
        \end{cases}\]
For $\pref \in \{\pct, \pca, \ppca \}$, we now have that
$Q \prec^\pref_i P$ iff $u^\pref_i(P) > u^\pref_i(Q)$. Interestingly, it can be shown that the set of Nash equilibria, in RCGBARs with any kind of attitude towards contention, remains unchanged when adopting these new definitions. 

\medskip
We showed in Theorem~\ref{thm:contention-averse-end-bags-always-NE} that a best-response dynamics can be applied from the empty profile and will converge to a Nash equilibrium. {\color{red} Showing it from an arbitrary profile remains a challenge.
We conjecture \label{conjecture} that RCGBARs with public contention-averse preferences are a class of (generalized ordinal) Potential Games~\cite{MONDERER1996124}, and thus that the best-response dynamics can be applied from any profile.
The formulation of preferences as utility function equips us with better tools to investigate this.
%
But settling the question has eluded us so far. We leave it for future work and we also hope that others will find the problem worth investigating.}\footnote{\input{conjecture-counterexample.tex}}

\medskip
Our models assume that the contributed resources are distributed non-cooperatively and without supervision.
In fact, single-mindedness makes the allocation of the contributed resources a nonissue in RCGs \emph{with contention-averse preferences}: distribute the resources available to the players who find the profile good, and leave the remaining resources unassigned.

However, an occasion will come to enrich the preferences beyond single-mindedness and let the players enjoy partial satisfaction when they receive some but not all of the resources in their objective.
Instead of having preferences about a raw profile $P$, the player's preferences could be raised over the result of the (fair, envy-free, efficient, etc) allocation of the resources~\cite{bouveretChap12} contributed in $\out(P)$.
A utility model in the vein of what is proposed above should help us in this direction.

\medskip

In general, the study of RCGs with other alternative preferences would be interesting. Having at disposition the notions of resources, satisfaction, contention, and private contention, it becomes rather tempting to think of other kind of preferences that might be fitting to model agents' attitudes in some application domain. Such alternative preferences could for instance consider the maximisation of potentially satisfied players, or the minimisation of players in private contention.

\medskip

This paper made no attempt at designing mechanisms to achieve `good' equilibria.
We are interested in using resource contribution games in problems of gamification.
Gamification refers to the broad application of game-design techniques in contexts that do not otherwise present game-like features \cite{gamificationEduBusiness}. 
Gamification aims at incentivizing an intended behavior by introducing rewards for specific tasks.
Rewards often present themselves as virtual resources such as achievement badges. Formally, they might be nothing more than distinguished tokens of resources.
In Example~\ref{ex:telecommunication}, we saw that the profile where all companies refrain from  providing any resources, $(\emptyset, \emptyset)$, is a Nash equilibrium. Also in Example~\ref{ex:grandfather}, we saw that the profile where everyone refrains from  providing any resources is a Nash equilibrium.
Nash equilibria in RCGs are in fact only faithful to the expected behaviours in resource-driven non-cooperative interactions.
It is a phenomenon observed also in experiments with public good games, demonstrating that voluntary private provisions of public goods typically remain underfunded~\cite{MARKISAAC198551}. This can be an undesirable behaviour that policy makers might be able to anticipate by using the analytical tools defined in this paper, and to avoid by using advanced gamification methods which must be the object of future research.

\section*{Acknowledgments}
The author wishes to thank the reviewers of this journal for pressing on crucial points.

\appendix

\section{Software implementation: examples, and proofs details}
\label{sec:implementation}

\lstdefinestyle{lststyle}{
  language=python,
  numbers=left,
  stepnumber=1,
  numbersep=10pt,
  tabsize=4,
  showspaces=false,
  showstringspaces=false,
  upquote
}

An implementation is available 
at \url{https://bitbucket.org/troquard/irgpy/}.

Suffixes \texttt{pct} stand for (parsimonious) contention-tolerant preferences ($\pct$), \texttt{pca} for (parsimonious public) contention-averse preferences ($\pca$), and \texttt{ppca} for (parsimonious) private contention-averse preferences ($\ppca$).

\subsection{Simple example}
We begin with a simple series of examples. Based on the same RCG (more precisely, an RCGBAR), we examine the differences in the sets of Nash equilibria when changing the kind of preferences.
\begin{example}[label=ex:implementation-exec]
  Player~$1$, Ann, is endowed with $\epsilon_1 = \{A,B\}$ and has the objective $\gamma_1 = A$. Player~$2$ and player~$3$, both identical Bob, are endowed with $\epsilon_2 = \epsilon_3 = \{A, B, B\}$ and have the objective $\gamma_2 = \gamma_3 = B$.

  The following script shows how the software can be used to determine the set of Nash equilibria in the game when considering (public) contention-averse preferences.

\begin{lstlisting}[breaklines,basicstyle=\ttfamily\scriptsize,style=lststyle]%\begin{verbatim}

$ python3
>>> from irg import Player, Game
>>> ann = Player("Ann", ['A','B'], ['A'])
>>> bob = Player("Bob", ['A', 'B', 'B'], ['B'])
>>> game = Game(ann, bob, bob)
>>> print(game)
Game: [
'(1, Player: Name: Ann. Endowment: {A, B}. Goal: {A}.)', 
'(2, Player: Name: Bob. Endowment: {A, B, B}. Goal: {B}.)', 
'(3, Player: Name: Bob. Endowment: {A, B, B}. Goal: {B}.)']
>>> for p in game.nash_generator_generic(game.prefers_pca):
...     print("Nash:", p)
... 
Nash: Profile: [(1, '{A}'), (2, '{}'), (3, '{B, B}')]
Nash: Profile: [(1, '{A}'), (2, '{B}'), (3, '{B}')]
Nash: Profile: [(1, '{A}'), (2, '{B, B}'), (3, '{}')]
Nash: Profile: [(1, '{A, B}'), (2, '{}'), (3, '{B}')]
Nash: Profile: [(1, '{A, B}'), (2, '{B}'), (3, '{}')]
\end{lstlisting}

\end{example}

\begin{example}[continues=ex:implementation-exec]
  Continuing Example~\ref{ex:implementation-exec}, we can see the effect on the set of Nash equilibria of a change from contention-averse preferences to contention-tolerant preferences.
  
\begin{lstlisting}[breaklines,basicstyle=\ttfamily\scriptsize,style=lststyle]%\begin{verbatim}

>>> for p in game.nash_generator_generic(game.prefers_pct):
...     print("Nash:", p)
... 
Nash: Profile: [(1, '{A}'), (2, '{}'), (3, '{B}')]
Nash: Profile: [(1, '{A}'), (2, '{B}'), (3, '{}')]
\end{lstlisting}
\end{example}

\begin{example}[continues=ex:implementation-exec]
  Continuing Example~\ref{ex:implementation-exec}, we can see the effect on the set of Nash equilibria of a change from (public) contention-averse preferences to private contention-averse preferences.
\begin{lstlisting}[breaklines,basicstyle=\ttfamily\scriptsize,style=lststyle]
>>> for p in game.nash_generator_generic(game.prefers_ppca):
...     print("Nash:", p)
... 
Nash: Profile: [(1, '{A}'), (2, '{}'), (3, '{B, B}')]
Nash: Profile: [(1, '{A}'), (2, '{B}'), (3, '{B}')]
Nash: Profile: [(1, '{A}'), (2, '{B, B}'), (3, '{}')]
\end{lstlisting}
\end{example}

\subsection{Back to Example~\ref{ex:bringyourownfood}}
We now demonstrate how the motivating example Example~\ref{ex:bringyourownfood} is implemented in the software.

\begin{example}[continues=ex:bringyourownfood]
In the first part of the example, with you, Bruno, and Carmen the game is formally defined as:
$\Gexone{.1} = (\{y, b, c\},\gamma_y = \bundle{\var{w}, \var{w}}, \gamma_b = \bundle{\var{w}, \var{r}, \var{o}}, \gamma_c = \bundle{\var{w}, \var{r}, \var{o}}, \epsilon_y = \{\var{w} \comb\var{w} \comb\var{w} \comb\var{w}\}, \epsilon_b = \{\var{r}, \var{r}, \var{r}\}, \epsilon_c = \{\var{o}, \var{o}\})$

\begin{lstlisting}[breaklines,basicstyle=\ttfamily\scriptsize,style=lststyle]
$ python3
>>> from irg import Player, Game
>>> you = Player("you", [('w','w','w','w')], 'ww')
>>> bruno = Player("bruno",  ['r','r','r'], 'wro')
>>> carmen = Player("carmen", ['o','o'], 'wro')
>>> game11 = Game(you, bruno, carmen)
>>> for p in game.nash_generator_generic(game.prefers_pca):
...     print("Nash (pca):", p)
...  
Nash (pca): Profile: [(1, "{('w', 'w', 'w', 'w')}"), (2, '{}'), (3, '{}')]
Nash (pca): Profile: [(1, "{('w', 'w', 'w', 'w')}"), (2, '{r, r}'), (3, '{o, o}')]
>>> for p in game11.nash_generator_generic(game11.prefers_ppca):
...     print("Nash (ppca):", p)
... 
Nash (ppca): Profile: [(1, "{('w', 'w', 'w', 'w')}"), (2, '{}'), (3, '{}')]
Nash (ppca): Profile: [(1, "{('w', 'w', 'w', 'w')}"), (2, '{r, r}'), (3, '{o, o}')]
>>> for p in game11.nash_generator_generic(game11.prefers_pct):
...     print("Nash (pct):", p)
... 
Nash (pct): Profile: [(1, "{('w', 'w', 'w', 'w')}"), (2, '{}'), (3, '{}')]
Nash (pct): Profile: [(1, "{('w', 'w', 'w', 'w')}"), (2, '{r}'), (3, '{o}')]
\end{lstlisting}

In the second part of the example, when Edward joins the party, the game is formally defined as:
$\Gexone{.2} = (\{y, b, c, e\},\gamma_y = \bundle{\var{w}, \var{w}}, \gamma_b = \bundle{\var{w}, \var{r}, \var{o}}, \gamma_c = \bundle{\var{w}, \var{r}, \var{o}}, \gamma_e = \bundle{\var{w}, \var{w}, \var{w}}, \epsilon_y = \{\var{w} \comb\var{w} \comb\var{w} \comb\var{w}\}, \epsilon_b = \{\var{r}, \var{r}, \var{r}\}, \epsilon_c = \{\var{o}, \var{o}\}, \epsilon_e = \emptyset)$

Continuing the previous execution, we simply add a player and define a new game with the four players.
\begin{lstlisting}[breaklines,basicstyle=\ttfamily\scriptsize,style=lststyle]
>>> edward = Player("edward", '', 'www')
>>> game12 = Game(you, bruno, carmen, edward)
>>> for p in game12.nash_generator_generic(game12.prefers_pca):
...     print("Nash (pca):", p)
... 
Nash (pca): Profile: [(1, '{}'), (2, '{}'), (3, '{}'), (4, '{}')]
>>> for p in game12.nash_generator_generic(game12.prefers_ppca):
...     print("Nash (ppca):", p)
... 
Nash (ppca): Profile: [(1, '{}'), (2, '{}'), (3, '{}'), (4, '{}')]
>>> for p in game12.nash_generator_generic(game12.prefers_pct):
...     print("Nash (pct):", p)
... 
Nash (pct): Profile: [(1, "{('w', 'w', 'w', 'w')}"), (2, '{}'), (3, '{}'), (4, '{}')]
Nash (pct): Profile: [(1, "{('w', 'w', 'w', 'w')}"), (2, '{r}'), (3, '{o}'), (4, '{}')]

\end{lstlisting}
\end{example}

\subsection{Back to Example~\ref{ex:telecommunication}}
We now demonstrate how our running example Example~\ref{ex:telecommunication} is implemented in the software.
\begin{example}[continues=ex:telecommunication]
The scenario was formalized into the RCG $\Gextwo = (\{a,b\},\gamma_a = \bundle{\var{3G}, \var{3G}, \var{4G}}, \gamma_b = \bundle{\var{3G}, \var{3G}, \var{4G}, \var{4G}}, \epsilon_a = \{\var{3G},\var{3G},\var{3G}\comb \var{3G}\}, \epsilon_b = \{\var{3G}, \var{4G}, \var{4G}\comb \var{4G}\})$.

\begin{lstlisting}[breaklines,basicstyle=\ttfamily\scriptsize,style=lststyle]
$ python3
>>> from irg import Player, Game
>>> A = Player("Company A", ['3G', '3G', ('3G', '3G')], ['3G', '3G', '4G'])
>>> B = Player("Company B", ['3G', '4G', ('4G', '4G')], ['3G', '3G', '4G' ,'4G'])
>>> game = Game(A, B)
>>> print(game)
Game: ["(1, Player: Name: Company A. Endowment: {3G, 3G, ('3G', '3G')}. 
Goal: {3G, 3G, 4G}.)", 
"(2, Player: Name: Company B. Endowment: {3G, 4G, ('4G', '4G')}. 
Goal: {3G, 3G, 4G, 4G}.)"]
>>> for p in game.nash_generator_generic(game.prefers_pca):
...     print("Nash:", p)
... 
Nash: Profile: [(1, '{}'), (2, '{}')]
Nash: Profile: [(1, "{3G, ('3G', '3G')}"), (2, "{3G, 4G, ('4G', '4G')}")]
Nash: Profile: [(1, "{3G, 3G, ('3G', '3G')}"), (2, "{4G, ('4G', '4G')}")]
>>> for p in game.nash_generator_generic(game.prefers_ppca):
...     print("Nash:", p)
... 
Nash: Profile: [(1, '{}'), (2, '{}')]
Nash: Profile: [(1, "{3G, ('3G', '3G')}"), (2, "{3G, 4G, ('4G', '4G')}")]
Nash: Profile: [(1, "{3G, 3G, ('3G', '3G')}"), (2, "{4G, ('4G', '4G')}")]
>>> for p in game.nash_generator_generic(game.prefers_pct):
...     print("Nash:", p)
... 
Nash: Profile: [(1, '{}'), (2, '{}')]
Nash: Profile: [(1, "{('3G', '3G')}"), (2, "{('4G', '4G')}")]
Nash: Profile: [(1, '{3G}'), (2, "{3G, ('4G', '4G')}")]
Nash: Profile: [(1, '{3G, 3G}'), (2, "{('4G', '4G')}")]
\end{lstlisting}
\end{example}

\subsection{Deviations explanations of Prop.~\ref{prop:no-nash}}
We can use the verbose mode of our implementation to repeat the proof of Prop.~\ref{prop:no-nash}.

\begin{example}
  We considered the RCG $G$, with two players~$1$ and $2$, where $\epsilon_1 = \epsilon_2 = \{A \comb B\}$, $\gamma_1 = A$, and $\gamma_2 = B \comb B$, and showed that $\NE^\pref(G) = \emptyset$ for every kind of preferences. This can be also verified programmatically as follows.

\begin{lstlisting}[breaklines,basicstyle=\ttfamily\scriptsize,style=lststyle]
$ python3
>>> from irg import Player, Game
>>> one = Player("one", [('A', 'B')], ['A'])
>>> two = Player("two", [('A', 'B')], ['B', 'B'])
>>> game = Game(one, two)
>>> for p in game.nash_generator_generic(game.prefers_pct, verbose=True):
...     print("Nash:", p) # there won't be any
... 
From Profile: [(1, '{}'), (2, '{}')] 
 one prefers Profile: [(1, "{('A', 'B')}"), (2, '{}')]
From Profile: [(1, '{}'), (2, "{('A', 'B')}")] 
 two prefers Profile: [(1, '{}'), (2, '{}')]
From Profile: [(1, "{('A', 'B')}"), (2, '{}')] 
 two prefers Profile: [(1, "{('A', 'B')}"), (2, "{('A', 'B')}")]
From Profile: [(1, "{('A', 'B')}"), (2, "{('A', 'B')}")] 
one prefers Profile: [(1, '{}'), (2, "{('A', 'B')}")]
>>> for p in game.nash_generator_generic(game.prefers_pca, verbose=True):
...     print("Nash:", p) # there won't be any
... 
From Profile: [(1, '{}'), (2, '{}')] 
 one prefers Profile: [(1, "{('A', 'B')}"), (2, '{}')]
From Profile: [(1, '{}'), (2, "{('A', 'B')}")] 
 two prefers Profile: [(1, '{}'), (2, '{}')]
From Profile: [(1, "{('A', 'B')}"), (2, '{}')] 
 two prefers Profile: [(1, "{('A', 'B')}"), (2, "{('A', 'B')}")]
From Profile: [(1, "{('A', 'B')}"), (2, "{('A', 'B')}")] 
 one prefers Profile: [(1, '{}'), (2, "{('A', 'B')}")]
>>> for p in game.nash_generator_generic(game.prefers_ppca, verbose=True):
...     print("Nash:", p) # there won't be any
... 
From Profile: [(1, '{}'), (2, '{}')] 
 one prefers Profile: [(1, "{('A', 'B')}"), (2, '{}')]
From Profile: [(1, '{}'), (2, "{('A', 'B')}")] 
 two prefers Profile: [(1, '{}'), (2, '{}')]
From Profile: [(1, "{('A', 'B')}"), (2, '{}')] 
 two prefers Profile: [(1, "{('A', 'B')}"), (2, "{('A', 'B')}")]
From Profile: [(1, "{('A', 'B')}"), (2, "{('A', 'B')}")] 
 one prefers Profile: [(1, '{}'), (2, "{('A', 'B')}")]
\end{lstlisting}
\end{example}

\subsection{Deviations explanations of Prop.~\ref{prop:no-nash-irgbar-ppca}}
We can use the verbose mode of our implementation to repeat the proof of Prop.~\ref{prop:no-nash-irgbar-ppca}.

\begin{example}\label{ex:irgbar-ppca}
We consider the RCGBAR $G$ where
$\epsilon_1 = \emptyset$, $\epsilon_2 = \emptyset$, $\epsilon_3 = \{B, B\}$, $\epsilon_4 = \{A, A\}$,
$\gamma_1 = \bundle{A,A,B,B}$,
$\gamma_2 = \bundle{A, B}$,
$\gamma_3 = B$,
and $\gamma_4 = A$. $\NE^\ppca(G)$ is empty.
\begin{lstlisting}[breaklines,basicstyle=\ttfamily\scriptsize,style=lststyle]
$ python3
>>> from irg import Player, Game
>>> one = Player("one", '', 'AABB')
>>> two = Player("two", '', 'AB')
>>> three = Player("three", 'BB', 'B')
>>> four = Player("four", 'AA', 'A')
>>> game = Game(one, two, three, four)
>>> for p in game.nash_generator_generic(game.prefers_ppca, verbose=True):
...     print("Nash: ", p)  # there won't be any
... 
From Profile: [(1, '{}'), (2, '{}'), (3, '{}'), (4, '{}')] 
 three prefers Profile: [(1, '{}'), (2, '{}'), (3, '{B}'), (4, '{}')]
From Profile: [(1, '{}'), (2, '{}'), (3, '{}'), (4, '{A}')] 
 three prefers Profile: [(1, '{}'), (2, '{}'), (3, '{B, B}'), (4, '{A}')]
From Profile: [(1, '{}'), (2, '{}'), (3, '{}'), (4, '{A, A}')] 
 four prefers Profile: [(1, '{}'), (2, '{}'), (3, '{}'), (4, '{A}')]
From Profile: [(1, '{}'), (2, '{}'), (3, '{B}'), (4, '{}')] 
 four prefers Profile: [(1, '{}'), (2, '{}'), (3, '{B}'), (4, '{A, A}')]
From Profile: [(1, '{}'), (2, '{}'), (3, '{B}'), (4, '{A}')] 
 three prefers Profile: [(1, '{}'), (2, '{}'), (3, '{}'), (4, '{A}')]
From Profile: [(1, '{}'), (2, '{}'), (3, '{B}'), (4, '{A, A}')] 
 three prefers Profile: [(1, '{}'), (2, '{}'), (3, '{}'), (4, '{A, A}')]
From Profile: [(1, '{}'), (2, '{}'), (3, '{B, B}'), (4, '{}')] 
 three prefers Profile: [(1, '{}'), (2, '{}'), (3, '{B}'), (4, '{}')]
From Profile: [(1, '{}'), (2, '{}'), (3, '{B, B}'), (4, '{A}')] 
 four prefers Profile: [(1, '{}'), (2, '{}'), (3, '{B, B}'), (4, '{}')]
From Profile: [(1, '{}'), (2, '{}'), (3, '{B, B}'), (4, '{A, A}')] 
 three prefers Profile: [(1, '{}'), (2, '{}'), (3, '{}'), (4, '{A, A}')]
\end{lstlisting}
\end{example}


\section*{Declarations}

\paragraph{Ethical Approval}
Not applicable.

\paragraph{Competing interests}
The author has no relevant financial or non-financial interests to disclose.

\paragraph{Authors' contributions}
N.T.\ is the sole author.
 
\paragraph{Funding}
The author is supported by a grant from the Free University of Bozen-Bolzano, \emph{Computations in Resource Aware Systems} (CompRAS).
 
\paragraph{Availability of data and materials}
An implementation of the models and algorithms presented in this paper is available at \url{https://bitbucket.org/troquard/irgpy/}.



\end{document}


%% file: conjecture-counterexample.tex
{\color{red} To see that the conjecture does not hold, it suffices to consider the very simple RCGBAR $(\{1,2,3\},  \gamma_1,\gamma_2,\gamma_3, \epsilon_1, \epsilon_2, \epsilon_3)$, where:
\begin{center}
\begin{tabular}{lll}
  \toprule
  $i$ & $\epsilon_i$ & $\gamma_i$\\
  \midrule
  $1$ & $\{A,B\}$ & $\bundle{A, B}$\\
  $2$ & $\{A,A\}$ & $\bundle{A, A}$\\
  $3$ & $\{B\}$ & $\bundle{B, B}$\\
  \bottomrule
\end{tabular}
\end{center}
and the profile $(\{B\}, \emptyset, \emptyset)$.
The following local improvements form a loop:
$(\{B\}, \{\}, \{\})
\stackrel{3}{\rightarrow} (\{B\}, \{\}, \{B\})
\stackrel{1}{\rightarrow} (\{A\}, \{\}, \{B\})
\stackrel{2}{\rightarrow} (\{A\}, \{A, A\}, \{B\})
\stackrel{3}{\rightarrow} (\{A\}, \{A, A\}, \{\})
\stackrel{2}{\rightarrow} (\{A\}, \{A\}, \{\})
\stackrel{1}{\rightarrow} (\{B\}, \{A\}, \{\})
\stackrel{2}{\rightarrow} (\{B\}, \{\}, \{\})
$.
}

%% file: contentious-paper.bbl
\begin{thebibliography}{10}

\bibitem{Airiau2014}
St{\'e}phane Airiau and Ulle Endriss.
\newblock Multiagent resource allocation with sharable items.
\newblock {\em Autonomous Agents and Multi-Agent Systems}, 28(6):956--985, 2014.

\bibitem{armstrong98interconnection}
Mark Armstrong.
\newblock Network interconnection in telecommunications.
\newblock {\em The Economic Journal}, 108(448):545--564, 1998.

\bibitem{aziz19}
H.~Aziz.
\newblock {Strategyproof multi-item exchange under single-minded dichotomous
  preferences}.
\newblock {\em Autonomous Agents and Multi-Agent Systems}, 34(3), 2020.

\bibitem{Bagnoli91}
Mark Bagnoli and Michael Mckee.
\newblock Voluntary contribution games: efficient private provision of public
  goods.
\newblock {\em Economic Inquiry}, 29(2):351--366, 1991.

\bibitem{BERGSTROM198625}
Theodore Bergstrom, Lawrence Blume, and Hal Varian.
\newblock On the private provision of public goods.
\newblock {\em Journal of Public Economics}, 29(1):25--49, 1986.

\bibitem{DBLP:conf/ecai/BonzonLLZ06}
Elise Bonzon, Marie{-}Christine Lagasquie{-}Schiex, J{\'{e}}r{\^{o}}me Lang,
  and Bruno Zanuttini.
\newblock {Boolean games revisited}.
\newblock In {\em Proceedings of the 17th European Conference on Artificial
  Intelligence (ECAI'06)}, volume 141 of {\em Frontiers in Artificial
  Intelligence and Applications}, pages 265--269. {IOS} Press, 2006.

\bibitem{bouveretChap12}
Sylvain Bouveret, Yann Chevaleyre, and Nicolas Maudet.
\newblock {\em Fair allocation of indivisible goods}, chapter~12, pages
  284--310.
\newblock In Brandt et~al. \cite{Brandt:2016:HCS:3033138}, 2016.

\bibitem{DBLP:journals/corr/BouyerBMU15}
Patricia Bouyer, Romain Brenguier, Nicolas Markey, and Michael Ummels.
\newblock {Pure Nash equilibria in concurrent deterministic games}.
\newblock {\em Log. Methods Comput. Sci.}, 11(2), 2015.

\bibitem{Brandt:2016:HCS:3033138}
Felix Brandt, Vincent Conitzer, Ulle Endriss, J{\'e}r\^{o}me Lang, and Ariel~D.
  Procaccia, editors.
\newblock {\em Handbook of Computational Social Choice}.
\newblock Cambridge University Press, New York, NY, USA, 2016.

\bibitem{brock1995interconnection}
W.~Gerald Brock.
\newblock {The economics of interconnection}.
\newblock Technical report, Teleport Communication Group, 1995.
\newblock Prepared for Teleport Communications Group.

\bibitem{10.1145/3467001.3467003}
V\'{e}ronique Bruy\`{e}re.
\newblock Synthesis of equilibria in infinite-duration: Games on graphs.
\newblock {\em ACM SIGLOG News}, 8(2):4--–29, May 2021.

\bibitem{Cerquides:2007:BLW:1625275.1625473}
Jes\'{u}s Cerquides, Ulle Endriss, Andrea Giovannucci, and Juan~A.
  Rodr\'{\i}guez-Aguilar.
\newblock Bidding languages and winner determination for mixed multi-unit
  combinatorial auctions.
\newblock In {\em Proceedings of the 20th International Joint Conference on
  Artifical Intelligence}, IJCAI'07, pages 1221--1226, San Francisco, CA, USA,
  2007. Morgan Kaufmann Publishers Inc.

\bibitem{DBLP:conf/csl/ChatterjeeMJ04}
Krishnendu Chatterjee, Rupak Majumdar, and Marcin Jurdzinski.
\newblock {On Nash equilibria in stochastic games}.
\newblock In {\em Proceedings of the 18th International Workshop on Computer
  Science Logic (CSL 2004)}, volume 3210 of {\em Lecture Notes in Computer
  Science}, pages 26--40. Springer, 2004.

\bibitem{DBLP:journals/informaticaSI/ChevaleyreDELLMPPRS06}
Yann Chevaleyre, Paul~E. Dunne, Ulle Endriss, J{\'{e}}r{\^{o}}me Lang, Michel
  Lema{\^{\i}}tre, Nicolas Maudet, Julian~A. Padget, Steve Phelps, Juan~A.
  Rodr{\'{\i}}guez{-}Aguilar, and Paulo Sousa.
\newblock Issues in multiagent resource allocation.
\newblock {\em Informatica (Slovenia)}, 30(1):3--31, 2006.

\bibitem{CoDiOuTr21aamas}
Rodica Condurache, Catalin Dima, Youssouf Oualhadj, and Nicolas Troquard.
\newblock Rational synthesis in the commons with careless and careful agents.
\newblock In {\em Proceedings of the 20th International Conference on
  Autonomous Agents and MultiAgent Systems} ({AAMAS} 2021), pages 368--376.
  International Foundation for Autonomous Agents and Multiagent Systems,
  2021. 

\bibitem{CoDiOuTr23aamas}
Rodica Condurache, Catalin Dima, Youssouf Oualhadj, and Nicolas Troquard.
\newblock Synthesis of resource-aware controllers against rational agents.
\newblock In {\em Proceedings of the 22nd International Conference on
  Autonomous Agents and MultiAgent Systems} ({AAMAS} 2023), pages 775--783.
  International Foundation for Autonomous Agents and Multiagent Systems,
  2023.
  
\bibitem{CramptonShohamSteinberg06}
Peter Cramton, Yoav Shoham, and Richard Steinberg, editors.
\newblock {\em Combinatorial auctions}.
\newblock MIT Press, 2006.

\bibitem{DBLP:journals/siamcomp/DaskalakisGP09}
Constantinos Daskalakis, Paul~W. Goldberg, and Christos~H. Papadimitriou.
\newblock {The complexity of computing a {Nash} equilibrium}.
\newblock {\em {SIAM} J. Comput.}, 39(1):195--259, 2009.

\bibitem{DBLP:conf/icalp/KeijzerKV20}
Bart de~Keijzer, Maria Kyropoulou, and Carmine Ventre.
\newblock Obviously strategyproof single-minded combinatorial auctions.
\newblock In {\em Proceedings of the 47th International Colloquium on Automata,
  Languages, and Programming ({ICALP} 2020)}, volume 168 of {\em LIPIcs}, pages
  71:1--71:17. Schloss Dagstuhl - Leibniz-Zentrum f{\"{u}}r Informatik, 2020.

\bibitem{Dunne201020}
Paul~E. Dunne, Sarit Kraus, Efrat Manisterski, and Michael Wooldridge.
\newblock {Solving coalitional resource games}.
\newblock {\em Artificial Intelligence}, 174(1):20--50, 2010.

\bibitem{ekeh74}
Peter~P. Ekeh.
\newblock {\em {Social Exchange Theory: The Two Traditions}}.
\newblock Harvard University Press, 1974.

\bibitem{10.1145/1007352.1007445}
Alex Fabrikant, Christos Papadimitriou, and Kunal Talwar.
\newblock {The complexity of pure {Nash} equilibria}.
\newblock In {\em Proceedings of the Thirty-Sixth Annual ACM Symposium on
  Theory of Computing}, STOC'04, pages 604--–612. Association for Computing
  Machinery, 2004.

\bibitem{doi:10.1177/000276427802100411}
Marcus Felson and Joe~L. Spaeth.
\newblock Community structure and collaborative consumption: A routine activity
  approach.
\newblock {\em American Behavioral Scientist}, 21(4):614--624, 1978.

\bibitem{DBLP:conf/tacas/FismanKL10}
Dana Fisman, Orna Kupferman, and Yoad Lustig.
\newblock Rational synthesis.
\newblock In Javier Esparza and Rupak Majumdar, editors, {\em Proceedings of
  the 22nd International Conference on Tools and Algorithms for the
  Construction and Analysis of Systems (TACAS)}, volume 6015 of {\em Lecture
  Notes in Computer Science}, pages 190--204. Springer, 2010.

\bibitem{Giovannucci2010}
Andrea Giovannucci, Jes{\'u}s Cerquides, Ulle Endriss, and Juan~A.
  Rodr{\'i}guez-Aguilar.
\newblock A graphical formalism for mixed multi-unit combinatorial auctions.
\newblock {\em Autonomous Agents and Multi-Agent Systems}, 20(3):342--368, 2010.

\bibitem{girard1987}
Jean-Yves Girard.
\newblock Linear logic.
\newblock {\em Theoretical Computer Science}, 50(1):1--101, 1987.

\bibitem{DBLP:journals/jair/GrandiGT19}
Umberto Grandi, Davide Grossi, and Paolo Turrini.
\newblock Negotiable votes.
\newblock {\em J. Artif. Intell. Res.}, 64:895--929, 2019.

\bibitem{GUTIERREZ201553}
Julian Gutierrez, Paul Harrenstein, and Michael Wooldridge.
\newblock Iterated boolean games.
\newblock {\em Information and Computation}, 242:53--79, 2015.

\bibitem{DBLP:journals/acta/GutierrezMPRSW21}
Julian Gutierrez, Aniello Murano, Giuseppe Perelli, Sasha Rubin, Thomas
  Steeples, and Michael~J. Wooldridge.
\newblock Equilibria for games with combined qualitative and quantitative
  objectives.
\newblock {\em Acta Informatica}, 58:585–-610, 2021.

\bibitem{harrenstein2015aamas}
Paul Harrenstein, Paolo Turrini, and Michael Wooldridge.
\newblock {Electric boolean games: redistribution schemes for resource-bounded
  agents}.
\newblock In {\em Proceedigns of the 14th International Conference on
  Autonomous Agents and Multi-agent Systems} ({AAMAS} 2015), pages 655--663.
  International Foundation for Autonomous Agents and Multiagent Systems,
  2015.

\bibitem{Harrenstein:2001:BG:1028128.1028159}
Paul Harrenstein, Wiebe van~der Hoek, John-Jules Meyer, and Cees Witteveen.
\newblock Boolean games.
\newblock In {\em Proceedings of the 8th Conference on Theoretical Aspects of
  Rationality and Knowledge}, TARK'01, pages 287--298. Morgan Kaufmann
  Publishers Inc., 2001.

\bibitem{ieong05}
Samuel Ieong, Robert McGrew, Eugene Nudelman, Yoav Shoham, and Qixiang Sun.
\newblock Fast and compact: A simple class of congestion games.
\newblock In {\em Proceedings of the 20th National Conference on Artificial
  Intelligence - Volume 2}, AAAI'05, pages 489-–494. AAAI Press, 2005.

\bibitem{DBLP:conf/atal/KothariSS04}
Anshul Kothari, Tuomas Sandholm, and Subhash Suri.
\newblock Solving combinatorial exchanges: optimality via a few partial bids.
\newblock In {\em Proceedings of the 3rd International Joint Conference on
  Autonomous Agents and Multiagent Systems ({AAMAS} 2004)}, pages 1418--1419.
  IEEE Computer Society,
  2004.

\bibitem{Kreps88}
David Kreps.
\newblock {\em Notes on the theory of choice}.
\newblock Underground Classics in Economics. Routledge, 1988.

\bibitem{MARKISAAC198551}
R.~{Mark Isaac}, Kenneth~F. McCue, and Charles~R. Plott.
\newblock Public goods provision in an experimental environment.
\newblock {\em Journal of Public Economics}, 26(1):51--74, 1985.

\bibitem{MARWELL1981295}
Gerald Marwell and Ruth~E. Ames.
\newblock {Economists free ride, does anyone else?: Experiments on the
  provision of public goods, IV}.
\newblock {\em Journal of Public Economics}, 15(3):295--310, 1981.

\bibitem{MILCHTAICH1996111}
Igal Milchtaich.
\newblock Congestion games with player-specific payoff functions.
\newblock {\em Games and Economic Behavior}, 13(1):111--124, 1996.

\bibitem{MONDERER1996124}
Dov Monderer and Lloyd~S. Shapley.
\newblock Potential games.
\newblock {\em Games and Economic Behavior}, 14(1):124--143, 1996.

\bibitem{Nash50}
John Nash.
\newblock {Equilibrium points in n-person games}.
\newblock {\em Proceedings of the National Academy of Sciences}, 36:48--49,
  1950.

\bibitem{Nash51}
John Nash.
\newblock {Non-cooperative games}.
\newblock {\em The Annals of Mathematics}, 54:286--295, 1951.

\bibitem{Ostrom90}
Elinor Ostrom.
\newblock {\em Governing the commons: the evolution of institutions for
  collective action}.
\newblock Cambridge University Press, 1990.

\bibitem{gamificationEduBusiness}
Torsten Reiners and Lincoln~C. Wood, editors.
\newblock {\em Gamification in education and business}.
\newblock Springer, Cham, 2015.

\bibitem{Rosenthal1973}
Robert~W. Rosenthal.
\newblock {A class of games possessing pure-strategy Nash equilibria}.
\newblock {\em International Journal of Game Theory}, 2(1):65--67, Dec 1973.

\bibitem{Su2020ACM-taas}
Zhaopin Su, Guofu Zhang, Feng Yue, Jindong He, Miqing Li, Bin Li, and Xin Yao.
\newblock Finding the largest successful coalition under the strict goal
  preferences of agents.
\newblock {\em ACM Trans. Auton. Adapt. Syst.}, 14(4), 2020.

\bibitem{sugden84}
Robert Sugden.
\newblock Reciprocity: The supply of public goods through voluntary
  contributions.
\newblock {\em The Economic Journal}, 94(376):772--787, 1984.

\bibitem{Tr16ijcai}
Nicolas Troquard.
\newblock Nash equilibria and their elimination in resource games.
\newblock In {\em Proceedings of the Twenty-Fifth International Joint
  Conference on Artificial Intelligence, {IJCAI} 2016}, pages 503--509. {AAAI}
  Press, 2016.

\bibitem{tr18aaai}
Nicolas Troquard.
\newblock Rich coalitional resource games.
\newblock In {\em Proceedings of the Thirty-Second {AAAI} Conference on
  Artificial Intelligence, {AAAI 2018}}, pages 1242--1249. {AAAI} Press, 2018.

\bibitem{10.1093/logcom/exaa031}
Nicolas Troquard.
\newblock {Individual resource games and resource redistributions}.
\newblock {\em Journal of Logic and Computation}, 30(5):1023--1062, 2020.

\bibitem{journals/amai/Troquard24}
Nicolas Troquard.
\newblock {Existence and verification of Nash equilibria in non-cooperative contribution games with resource contention}.
\newblock {\em Ann. Math. Artif. Intell.}, 92(2):317--353, 2024.

\bibitem{WOOLDRIDGE2013418}
Michael Wooldridge, Ulle Endriss, Sarit Kraus, and Jérôme Lang.
\newblock Incentive engineering for boolean games.
\newblock {\em Artificial Intelligence}, 195:418--439, 2013.

\bibitem{DBLP:conf/aaai/WooldridgeGHMPT16}
Michael~J. Wooldridge, Julian Gutierrez, Paul Harrenstein, Enrico Marchioni,
  Giuseppe Perelli, and Alexis Toumi.
\newblock {Rational verification: from model checking to equilibrium checking}.
\newblock In {\em Proceedings of the Thirtieth {AAAI} Conference on Artificial
  Intelligence, AAAI 2016}, pages 4184--4191. {AAAI} Press, 2016.

\bibitem{yamagishiCook93}
Toshio Yamagishi and Karen S. Cook
\newblock {Generalized exchange and social dilemmas}.
\newblock {\em Social Psychology Quarterly}, 56(4):235–--248, 1993.

\end{thebibliography}
